\documentclass[aps,pre,showpacs,10pt,twocolumn,superscriptaddress]{revtex4-1}

\usepackage{amsmath,amsthm}
\usepackage{amssymb}
\usepackage{amscd}
\usepackage{wasysym}
\usepackage[ansinew]{inputenc}
\usepackage[T1]{fontenc}
\usepackage{ae,aecompl}
\usepackage{hyperref}
\usepackage{sidecap}
\usepackage{color}

\usepackage[pdftex]{graphicx}
\DeclareGraphicsExtensions{.pdf}

\newcommand{\be}{\begin{equation}}
\newcommand{\ee}{\end{equation}}
\newcommand{\bea}{\begin{eqnarray}}
\newcommand{\eea}{\end{eqnarray}}
\newcommand{\nn}{\nonumber}
\renewcommand{\vec}[1]{\boldsymbol{#1}}
\newcommand{\BB}{\mathcal{B}}

\newtheorem{thm}{Theorem}
\newtheorem{defn}[thm]{Definition}
\newtheorem{prop}{Proposition}

\newtheorem{lemma}{Lemma}

\begin{document}
\bibliographystyle{apsrev}

\title{Non-Local Impact of Link Failures in Linear Flow Networks} %

\author{Julius Strake}
\thanks{JS and FK contributed equally to this work.}
\affiliation{Forschungszentrum J\"ulich, Institute for Energy and Climate Research -
	Systems Analysis and Technology Evaluation (IEK-STE),  52428 J\"ulich, Germany}
\affiliation{University of Cologne, Institute for Theoretical Physics, 
		50937 K\"oln, Germany}
		
\author{Franz Kaiser}
\thanks{JS and FK contributed equally to this work.}
\affiliation{Forschungszentrum J\"ulich, Institute for Energy and Climate Research -
	Systems Analysis and Technology Evaluation (IEK-STE),  52428 J\"ulich, Germany}
\affiliation{University of Cologne, Institute for Theoretical Physics, 
		50937 K\"oln, Germany}

		\author{Farnaz Basiri}
\affiliation{Forschungszentrum J\"ulich, Institute for Energy and Climate Research -
	Systems Analysis and Technology Evaluation (IEK-STE),  52428 J\"ulich, Germany}

\author{Henrik Ronellenfitsch}
\affiliation{Department of Mathematics, Massachusetts Institute of Technology, Cambridge, MA 02139, U.S.A.}

\author{Dirk Witthaut}
\affiliation{Forschungszentrum J\"ulich, Institute for Energy and Climate Research -
	Systems Analysis and Technology Evaluation (IEK-STE),  52428 J\"ulich, Germany}
\affiliation{University of Cologne, Institute for Theoretical Physics, 
		50937 K\"oln, Germany}

\date{\today }

\begin{abstract}
The failure of a single link can degrade the operation of a supply network up to the point of complete collapse. Yet, the interplay between network topology and locality of the response to such damage is poorly understood. Here, we study how topology affects the redistribution of flow after the failure of a single link in linear flow networks with a special focus on power grids. In particular, we analyze the decay of flow changes with distance after a link failure and map it to the field of an electrical dipole for lattice-like networks. The corresponding inverse-square law is shown to hold for all regular tilings. For sparse networks, a long-range response is found instead. In the case of more realistic topologies, we introduce a rerouting distance, which captures the decay of flow changes better than the traditional geodesic distance. Finally, we are able to derive rigorous bounds on the strength of the decay for arbitrary topologies that we verify through extensive numerical simulations. Our results show that it is possible to forecast flow rerouting after link failures to a large extent based on purely topological measures and that these effects generally decay with distance from the failing link. They might be used to predict links prone to failure in supply networks such as power grids and thus help to construct grids providing a more robust and reliable power supply. 
\end{abstract}

\maketitle

\section{Introduction}
\label{sec:intro}
The robust operation of supply networks is essential for the function of complex systems across scales and disciplines. Almost all of our technical and economical infrastructure depends on the reliable operation of the electric power grid~\cite{Marr08,Vleu10}. Biological organisms distribute water and nutrients via their vascular networks, for instance in plant leaves~\cite{Sack2013}, the human and animal circulatory system~\cite{Pittmann2011}, or in protoplasmic veins of certain slime molds~\cite{Tero2010}. Huge amounts of money and assets are exchanged through a complex financial network~\cite{Schw09}. Structural damages to such networks can have catastrophic consequences such as a stroke, a power outage or a major economic crisis.

In power grids, large scale outages are typically triggered by the failure of a single transmission or generation element~\cite{Pour06,Yang2017,Rohden2016,16redundancy,Schafer2018}. The outages in the United States in 2003, Italy in 2003 and Western Europe in 2006 are very well documented and provide a detailed insight into the dynamics of a large scale network failure~\cite{USout03,Ande05,UCTE07}. Each outage was triggered by the loss of a transmission line during a period of high grid load. Subsequently, the power flows were rerouted, causing secondary overloads and eventually a cascade of failures. In these three examples, the cascades propagated mostly locally -- overloads took place in the proximity of previous failures. However, this is not necessarily the case during power outages (see, e.g.~\cite{USout96}), raising the question of how networks flows are rerouted after failures~\cite{Dobs16,Hine17, Labavic2014,Jung2015,Kett15,15nonlocal,Witthaut13,tyloo2018}.

In biological distribution networks, robustness against link failure is a critical prerequisite that guards against potentially life-threatening events such as stroke~\cite{Schaffer2006} or embolism~\cite{Roth-Nebelsick2001,Sack2008}, but also to function efficiently in the presence of fluctuations~\cite{Nesti2018,Sack2013,Kati10}. Thus, biological networks are often (but not in all cases, such as in the penetrating arterioles of the cortical vasculature~\cite{Shih2013}) endowed with highly resilient, redundant topologies that optimize rerouting of flow in case of link failure to the network~\cite{Kati10} and are generated through adaptive developmental mechanisms~\cite{Ronellenfitsch2017b}. For the understanding of such life-threatening conditions it is therefore crucial to investigate the behavior of the vascular network in the case of failure.

To understand the vulnerability of networks, we here provide a detailed analysis of the impact of link failures in linear flow networks. We focus on how the network topology determines the overall network response as well as the spatial flow rerouting. We consider linear supply network models, where the flow between two adjacent nodes is proportional to the difference of the nodal potential, pressure or voltage phase angle. Linear models are applied to hydraulic networks~\cite{Hwan96}, vascular networks of plants and animals~\cite{Kati10,Hu2013,Corson2010,Ronellenfitsch2016,Gavrilchenko18}, economic input-output networks \cite{Miller2009} as well as electric power grids~\cite{Grai94,Wood14,Purc05,Hert06,Horsch18,Manik17}. The linearity allows to obtain several rigorous bounds for flow rerouting in general network topologies and to solve special cases fully analytically.

\section{Linear flow networks}

Consider a network consisting of $N$ nodes that are connected to each other via lines denoted by $(m,n)$ for a line going from node $m$ to node $n$. Assign a potential or phase angle $\theta_m\in\mathbb{R}$ to each node $m$ in the network.
Then we assume the flow $F_{m \rightarrow n}$ between nodes $m$ and $n$ connected via line $(m,n)$ to be \emph{linear} in the potential drop along the line
\be
    F_{m \rightarrow n} = b_{mn} (\theta_m - \theta_n).
    \label{eq:flow-line}
\ee
Here, $b_{mn}=b_{nm}$ is the transmission capacity assigned to the line $(m,n)$ that describes its ability to carry flow. This equation may for example be used to describe hydraulic networks~\cite{Hwan96,diaz2016} or vascular networks of plants~\cite{Kati10}, where the $\theta_n$ denotes the pressure at some node $n$ and the capacity $b_{mn}$ scales with the diameter of a pipe or vein.
Our main focus will be its application to electric power engineering, where this linear approximation of the power flow equations is referred to as the DC approximation~\cite{Wood14,Purc05,Hert06}. In this case, $F_{m \rightarrow n}$ refers to the flow of real power along a transmission line $(m,n)$, $\theta_n$ is the voltage phase angle at node $n$ and $b_{mn}$ is proportional to the line's susceptance. 

In addition to that, we assume that Kirchhoff's current law holds at the nodes of the network which states that the inflows and outflows at any node $m$ balance%
\be
    \sum_{n=1}^N F_{m \rightarrow n} = P_m,
    \label{eq:continuity}
\ee
where the right-hand-side denotes the inflow ($P_m>0$) or outflow ($P_m < 0$) at node $m$, commonly called the `power injection' in power engineering. Equations (\ref{eq:flow-line}) and (\ref{eq:continuity}) fully describe the state and the flow of the network once the line parameters $b_{mn}$ and the injections $P_m$ are given.

These equations may be conveniently written using a vectorial notation. Define the vector $\vec \theta = (\theta_1,\ldots,\theta_N)^\top \in \mathbb{R}^N$ of the nodal potentials or voltage phase angles and the vector $\vec P = (P_1,\ldots,P_N)^\top \in \mathbb{R}^N$ of nodal injections. Here and in the following sections, the superscript `$\top$' denotes the transpose of a vector or matrix. We further label all lines in the grid by $\ell = 1,\ldots,L$ and summarize all line flows in a vector $\vec F = (F_1,\ldots,F_L)^\top \in \mathbb{R}^L$.
Equation~(\ref{eq:flow-line}) may then be rewritten as
\be
    \vec F = \vec B_d \vec K^\top \vec \theta,\nonumber
\ee
where $\vec B_d \in \mathbb{R}^{L \times L}$ is a diagonal matrix containing the capacities $b_\ell$ of all edges. Furthermore, we defined the node-edge incidence matrix $\vec K \in \mathbb{R}^{N \times L}$. To define this matrix in an undirected graph, one typically fixes an arbitrary orientation of the graph's edges such that its components read%
\be
   K_{n,\ell} = \left\{
   \begin{array}{r l}
      1 & \; \mbox{if line $\ell$ starts at node $n$},  \\
      - 1 & \; \mbox{if line $\ell$ ends at node $n$},  \\
      0     & \; \mbox{otherwise}.
  \end{array} \right.
  \nonumber%
\ee

The node-edge incidence matrix also relates the injections to the flows incident at a node. More specifically, Kirchhoff's current law~(\ref{eq:continuity}) may be rewritten as follows%
\begin{equation}
  \vec P = \vec K \vec F = \vec K \vec B_d \vec K^\top \vec \theta = \vec B \vec \theta.
  \label{eq:DCapprox}
\end{equation}
Here, we defined the matrix $\vec B = \vec K \vec B_d \vec K^\top \in \mathbb{R}^{N \times N}$ commonly referred to as the nodal susceptance matrix in power engineering. Mathematically, $\vec B$ is a weighted Laplacian matrix~\cite{Newm10,merris1994} with components
\begin{equation}
  b_{mn} = \left\{
   \begin{array}{lll}
   \displaystyle\sum \nolimits_{\ell \in \Lambda_m} b_{\ell} &  \mbox{if } m = n; \\ [2mm]
     - b_{\ell} & \mbox{if }  m \mbox{ is connected to } n \mbox{ by } \ell.
   \end{array} \right. \nonumber%
\end{equation}
Here, $\Lambda_m$ is the set of lines which are incident to $m$.

\section{Algebraic description and analysis of line outages}

An important question in network security analysis is how the flows in the network change if a line fails. Denoting by $F_{\ell}$ the initial flow of the failing line $\ell \hat = (r,s)$, the flow change $\Delta F_e$ at a transmission line $e \hat = (m,n)$ is written as
\be
    \Delta F_e = \mbox{LODF}_{e,\ell} \, F_{\ell} .\nonumber
\ee
Adopting the language of power system security analysis~\cite{Grai94,Wood14},
we call the factor of proportionality the line outage distribution factor (LODF). In the following, we present two alternative derivations as well as a physical interpretation of the linear flow rerouting problem.

\subsection{Self-consistent derivation of line outage distribution factors}

To derive an explicit expression for the LODFs one generally starts with a related problem. Consider an increase of the real power injection at node $r$ and a corresponding decrease at node $s$ by the amount $\Delta P$.
The new vector of real power injections is then given by
\be
  \hat{\vec P} = \vec P + \Delta P\, \vec \nu_{rs},
  \label{eqn:real-power-balance}
\ee
where the components of $\vec \nu_{rs} \in \mathbb{R}^N$  are $+1$ at position $r$, $-1$ at position $s$ and zero otherwise. Here and in the following, we use a hat to indicate the state of the network \emph{after} a line outage or a similar change of the topology. The change of the real power injections causes a change of the nodal voltage angles,
\be
  \vec \psi = \hat{\vec \theta} - \vec \theta =  \Delta P \, \vec B^{\dagger} \vec \nu_{rs},
  \nonumber%
\ee
where $\vec B^{\dagger}$ denotes the Moore-Penrose pseudo-inverse of the Laplacian matrix $\vec B$. Hence, the real power flows change by
\begin{align}
    \Delta F_{mn} &= b_{mn} (\psi_m - \psi_n) =
      \underbrace{b_{mn} \vec \nu_{mn}^\top \vec B^{\dagger} \vec \nu_{rs}}_{
          =: {\rm PTDF}_{(m,n),r,s}} \Delta P.
    \label{eq:ptdf-definition}
\end{align}
The factor of proportionality is referred to as the power transfer distribution factor (PTDF).

The LODFs can be expressed by PTDFs in the following way \cite{Wood14}. To consistently model the outage of line $(r,s)$, one assumes that the line is disconnected from the grid by circuit breakers and that some fictitious real power $\Delta P$ is injected at node $s$ and taken out at node $r$. The entire flow over the line $(r,s)$ after the opening thus equals the fictitious injections $\hat F_{rs} = \Delta P$. Using PTDFs, we also know that
\begin{align}
   \hat F_{rs} = F_{rs}  + \mbox{PTDF}_{(r,s),r,s} \, \Delta P \, .\nonumber
\end{align}
Substituting $\hat F_{rs} = \Delta P$ and solving for $\Delta P$ yields
\begin{align}
 \Delta P = \hat F_{rs} = \frac{F_{rs}}{1 - \mbox{PTDF}_{(r,s),r,s}} \, .\nonumber
\end{align}
The change of real power flows of all other lines is given by
Equation (\ref{eq:ptdf-definition}) such that we
finally obtain
\be
    \mbox{LODF}_{(mn),(rs)} =\frac{\mbox{PTDF}_{(m,n),r,s}}{1 - \mbox{PTDF}_{(r,s),r,s}} \, .
    \label{eqn:lodf1}
\ee
For consistency, one usually defines the $\mbox{LODF}$ for the failing line as follows $\mbox{LODF}_{(rs),(rs)} = -1$. In addition to that, we exclude cases where the failing line is a bridge, i.e., a line whose removal disconnects the graph, from our analysis in the following sections.

\subsection{Algebraic derivation of line outage distribution factors}

The LODFs can also be obtained in a purely algebraic way without any self-consistency argument~\cite{Guo09}. As the line $\ell \hat = (s,r)$ fails, the nodal susceptance matrix of the network changes as
\begin{align}
   \vec B \rightarrow  \hat{\vec B} &= \vec B + \Delta \vec B  \nonumber \\
   \Delta \vec B &= B_{rs} \vec \nu_{rs} \vec \nu_{rs}^\top,
   \label{eq:def-Vsingle}
\end{align}
which causes a change of the nodal potentials or voltage phase angles respectively,
\be
    \vec \theta \rightarrow  \hat{\vec \theta} = \vec \theta + \vec \psi.\nonumber
\ee
Equation (\ref{eq:DCapprox}) for the perturbed grid now reads
\be
   (\vec B + \Delta \vec B) (\vec \theta + \vec \psi) = \vec P.
   \nonumber%
\ee
Subtracting Equation (\ref{eq:DCapprox}) for the unperturbed grid, we see that the change of the voltage angles is given by
\begin{align}
    \vec \psi &= - (\vec B + \Delta \vec B)^{\dagger}  \Delta \vec B \, \vec \theta \nn \\
           &=   (\vec B + \Delta \vec B)^{\dagger} \,  \vec \nu_{rs}  F_{rs} .
   \label{eqn:DeltaTheta}
\end{align}
The change of flows after the outage of line $(r,s)$ and thus the LODFs are calculated from the change of the voltage angles which yields
\begin{align}
  \Delta F_{mn} &= b_{mn} (\psi_{m} - \psi_{n}) \nn  \\
                    &=  b_{mn} \vec \nu_{mn}^\top \vec \psi \nn \\
                    &=  b_{mn} \vec \nu_{mn}^\top (\vec B + \Delta \vec B)^{\dagger} \vec \nu_{rs} F_{rs}.
                    \label{eqn:DeltaF}
\end{align}
In principle, we could now use these equations to calculate the flow changes and the LODFs. However, this would require to invert the matrix $\hat{\vec B}=\vec B + \Delta \vec B$ separately for every possible line $(s,r)$ in the grid, which is impractical. Nevertheless, we can simplify the expression using the Woodbury matrix identity,
\bea
    && (\vec B + B_{rs} \vec \nu_{rs} \vec \nu_{rs}^\top)^{\dagger} = \nn \\
   && \qquad \vec B^{\dagger} -
      \vec B^{\dagger} \vec  \nu_{rs}  (B_{rs}^{\dagger} +\vec \nu_{rs}^\top   \vec \nu_{rs}  )^{\dagger}
      \vec \nu_{rs}^\top \vec B^{\dagger}.\nonumber
\eea
Thus we obtain
\be
     (\vec B + B_{rs} \vec \nu_{rs} \vec \nu_{rs}^\top)^{\dagger} \vec \nu_{rs} =
   (1 + B_{rs} \vec \nu_{rs}^\top \vec B^\dagger  \vec \nu_{rs})^{-1} \,  \vec B^\dagger  \vec \nu_{rs},
   \label{eq:Woodbury2}
\ee
such that the flow change (\ref{eqn:DeltaF}) reads
\be
    \Delta F_{mn} =
     \frac{b_{mn} \vec \nu_{mn}^\top \vec B^\dagger  \vec \nu_{rs}}{
       1 - b_{rs} \vec \nu_{rs}^\top \vec B^\dagger  \vec \nu_{rs}} \, \times F_{rs} \, \nonumber,
\ee
which is identical to Equation~(\ref{eqn:lodf1}) obtained using the standard approach.

\subsection{Electrostatic interpretation}

A deeper physical insight into the network flow rerouting problem is obtained by the
analogy to discrete electrostatics. Equation (\ref{eqn:DeltaTheta}) can be rearranged
into  a linear set of equations for the change of the nodal potentials 
\be
   \hat { \vec B}  \vec \psi = F_{rs} \vec \nu_{sr}  \,  .
   \label{eq:Poisson-1}
\ee
Here, $\hat {\vec B}$ is the Laplacian of the grid \emph{after} the failure, i.e.~the grid where line $(r,s)$ has been removed. Alternatively, we can formulate the equation in terms of the original network topology, substituting Equation (\ref{eq:Woodbury2}) into Equation (\ref{eqn:DeltaTheta}). This yields the linear set of equations
\be
    \vec B \, \vec \psi =  \vec q
     \label{eq:Poisson}
\ee
with the dipole source
\be
    \vec q = 
     (1 - b_{rs} \vec \nu_{rs}^\top \vec B^\dagger  \vec \nu_{rs})^{-1} F_{rs} \vec \nu_{sr} \, .
\ee
As noted before, $\vec B$ and $\hat{\vec B}$ are Laplacian matrices and the right-hand side of both equations (\ref{eq:Poisson}) and (\ref{eq:Poisson-1}) are non-zero only at positions $r$ and $s$ with opposite sign. Hence, these equations are \emph{discrete Poisson equations} with a dipole source and $\vec \psi$ is a dipole potential, see Ref.~\cite[ch. 15]{Norman97} for a detailed analysis of this equation. The main complexity of the line outage problem thus arises from the network topology encoded in the Laplacian $\vec B$, which can be highly irregular.

The two equations (\ref{eq:Poisson}) and (\ref{eq:Poisson-1}) yield the same potential $\vec \psi$, but are formulated on different topologies -- either on the original network topology or the topology after the outage. To compare the impact of different failures it is beneficial to use the original topology, such that only the dipole inhomogeneity differs -- not the electrostatic problem itself.
Then, the strength of the dipole depends on the network topology via the prefactor $(1 - b_{rs} \vec \nu_{rs}^\top \vec B^\dagger  \vec \nu_{rs})^{-1}$.

Using the analogy to electrostatics we can solve the flow rerouting problem for regular network topologies (section \ref{sec:square}) and provide some general rigorous results (section \ref{sec:rigorous}). To understand flow rerouting in networks with complex topologies, we thus have to account for the spatial spreading pattern described by $\vec B^\dagger$ (see section \ref{sec:decay-num}) as well as the dipole strength, which quantifies the gross response of the grid (see section \ref{sec:dipole}).

\section{Failures in regular networks and the continuum limit}
\label{sec:square}

\begin{figure}[tb]
\begin{center}
\includegraphics[width=1.\columnwidth]{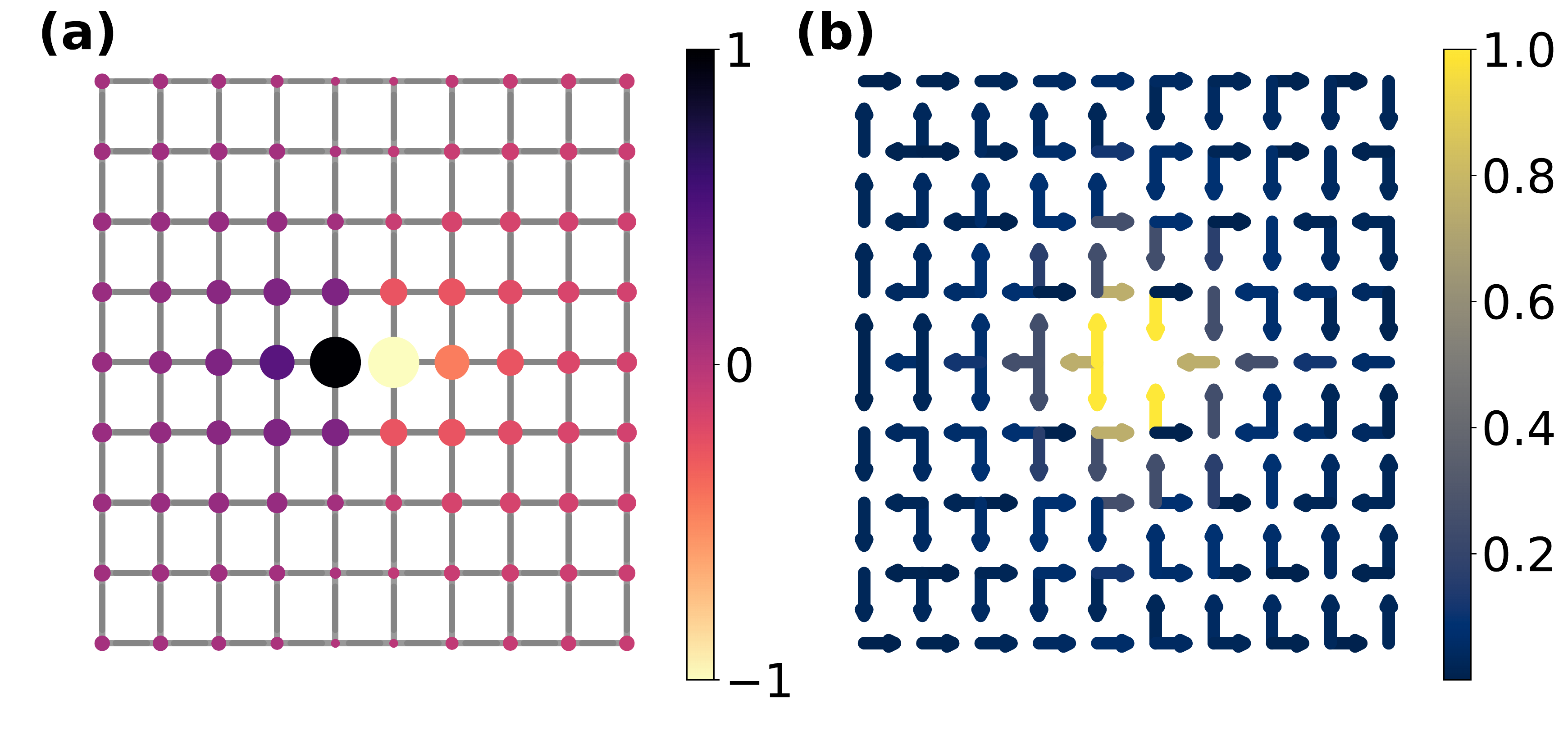}
\end{center}
\caption{
\label{fig:lattice}
Impact of a link failure in a homogeneous square lattice. (a) Normalized change of the nodal potentials $\psi_n$ for a network with uniform edge weights for a single failing link located in the center of the network. The size of the nodes as well as the colorcode represent the strength of the change in potential. The change is strongest close to the failing link and decays with distance. %
(b) Normalized change of the link flows $\Delta F_{nm}$ for the same topology. Arrows and color represent direction and strength of flow changes, respectively. The pattern corresponds to the one produced by an electrostatic dipole.}
\end{figure}

To obtain a first insight into the spatial aspects of flow rerouting, we consider an elementary example admitting a solution in the continuum limit. Consider a regular square lattice embedded in a plane as depicted in Figure~\ref{fig:lattice} and studied in a slightly different form in Ref.~\cite{Jozsef2002}. All nodes are labeled by their positions $\vec r = (x,y)^\top$ in this two-dimensional embedding and the lattice spacing is denoted as $h$. We introduce continuous functions $\psi$ and $b$ such that $\psi(x,y)$ is the potential of the node at $(x,y)$ and $b(x+h/2,y)$ is the weight of the link connecting the two nodes at $(x,y)$ and $(x+h,y)$. The left-hand side of the Poisson equation (\ref{eq:Poisson}) evaluated at position $(x,y)$ reads
\begin{align}
    & (\vec B \psi)(x,y)  \nn\\
     &\; = b(x+h/2,y) \left[ \psi(x,y) - \psi(x+h,y) \right] \nn \\
     &\quad + b(x-h/2,y) \left[ \psi(x,y) - \psi(x-h,y) \right] \nn \\
     &\quad + b(x,y+h/2) \left[ \psi(x,y) - \psi(x,y+h) \right] \nn \\
     &\quad + b(x,y-h/2) \left[ \psi(x,y) - \psi(x,y-h) \right] \nn \\
    &\; = -h^2 \nabla \cdot \left(b(x,y) \nabla \psi \right) + \mathcal{O}(h^3).
    \label{eq:PLHS-cont}
\end{align}
Here, we made use of the fact that the components of the gradient $\nabla\psi=(\partial_x \psi,\partial_y \psi)^\top$ may be expressed as 
\begin{align*}
    \frac{\partial \psi(x,y)}{\partial x}=\lim\limits_{h\rightarrow 0}\frac{\psi(x+h,y)-\psi(x,y)}{h},
\end{align*}
but did not take the limit yet.
The derivative with respect to $y$ may be calculated analogously.

Before we proceed to the right-hand side, we remark that the flow changes $\Delta \vec{F}$ according to equation (\ref{eqn:DeltaF}) are given by
given by
\begin{align}
   \Delta \vec{F}_x(x+h/2,y) &= b(x+h/2,y) ( \psi(x+h,y) - \psi(x,y)  )  \nn \\
   \Delta \vec{F}_y(x,y+h/2) &= b(x,y+h/2) ( \psi(x,y+h) - \psi(x,y)  ) \nn
\end{align}
such that we obtain in the continuum limit $h\rightarrow 0$,
\be
   \Delta \vec F(x,y) =b(x,y) \nabla \psi(x,y).
   \label{eq:flow-cont}
\ee
Note that the expression $\Delta \vec{F}$ refers to the change in flow due to the link failure here and should not be confused with the continuous Laplace operator.

The right-hand side of the discrete Poisson equation (\ref{eq:Poisson}) may be calculated similarly noting that only two
nodes contribute with opposite signs. Let us assume that the failing link is parallel to the
$x$-axis connecting nodes $r$ and $s$ located at $\vec{r}_r=(x_r,y_r)^\top$ and $\vec{r}_s=(x_s,y_s)^\top=(x_r+h,y_r)^\top$. %
The discrete version of the right-hand side reads
\begin{align*}
    \vec q = \frac{F_{rs}}{1 - b_{rs} \vec \nu_{rs}^\top \vec{B}^\dagger \vec \nu_{rs}} \vec \nu_{rs}.
\end{align*}
We will now derive the continuum version of this equation. First, the flow on the failing link before the outage $F_{rs}$ may be calculated as
\begin{align*}
    F_{rs} &\equiv b(x_r+h/2,y_r) (\theta(x_r+h,y_r) - \theta(x_r,y_r))  \nn \\
    &= h b(x_r,y_r) \frac{\partial \theta}{\partial x} + \mathcal{O}(h^2) \\
    &= h \vec F_x(x_r,y_r) + \mathcal{O}(h^2),
\end{align*}
where $\vec F(x,y) = b(x,y) \vec \nabla \theta(x,y)$ is the continuum flow before the outage.
Second, the vector $\vec \nu_{rs}$ can be formally interpreted in terms of the two-dimensional delta function $\delta(x,y)$ and reads for the given link failure
\begin{align*}
\vec \nu_{rs} &\equiv \delta(x-x_r+h,y-y_r) - \delta(x-x_r,y-y_r) \\
&= h \frac{\partial\delta(x-x_r,y-y_r)}{\partial x} + \mathcal{O}(h^2).
\end{align*}
Finally, let us assume that a continuum version of the Green's function $\vec B^\dagger$ exists. Then the denominator may be calculated as
\begin{align*}
  b_{rs} \vec \nu_{rs}^\top \vec{B}^\dagger \vec \nu_{rs} &\equiv  h^2\cdot b(x_r+h/2,y_r)\Big(\int\partial_y\delta(x-x_r,y-y_r)\\
  & \qquad b^\dagger(x,y) \partial_x\delta(x-x_r,y-y_r)\,\mathrm{d}x\, \mathrm{d}y\Big)\\
 &= h^2 b(x_r,y_r) \frac{\partial^2 b^\dagger(x_r,y_r)}{\partial x \partial y} + \mathcal{O}(h^3),
\end{align*}
where $b^\dagger(x,y)$ is the aforementioned continuum version.

Thus, in total we obtain after expanding the entire right-hand side to lowest order in the continuum limit
\begin{align}
    q(x,y) = h^2 \vec F(x_r,y_r)^\top \vec \nabla \delta(x-x_r,y-y_r) + \mathcal{O}(h^3).
   \label{eq:PRHS-cont}
\end{align}
Here, $\vec F(x_r,y_r)$ is assumed to be parallel to the dipole axis, i.e. the direction of the link failure, which is either the $x$- or the $y$-direction for the given setting.

\begin{figure*}[t!]
    \begin{center}
        \includegraphics[width=\textwidth]{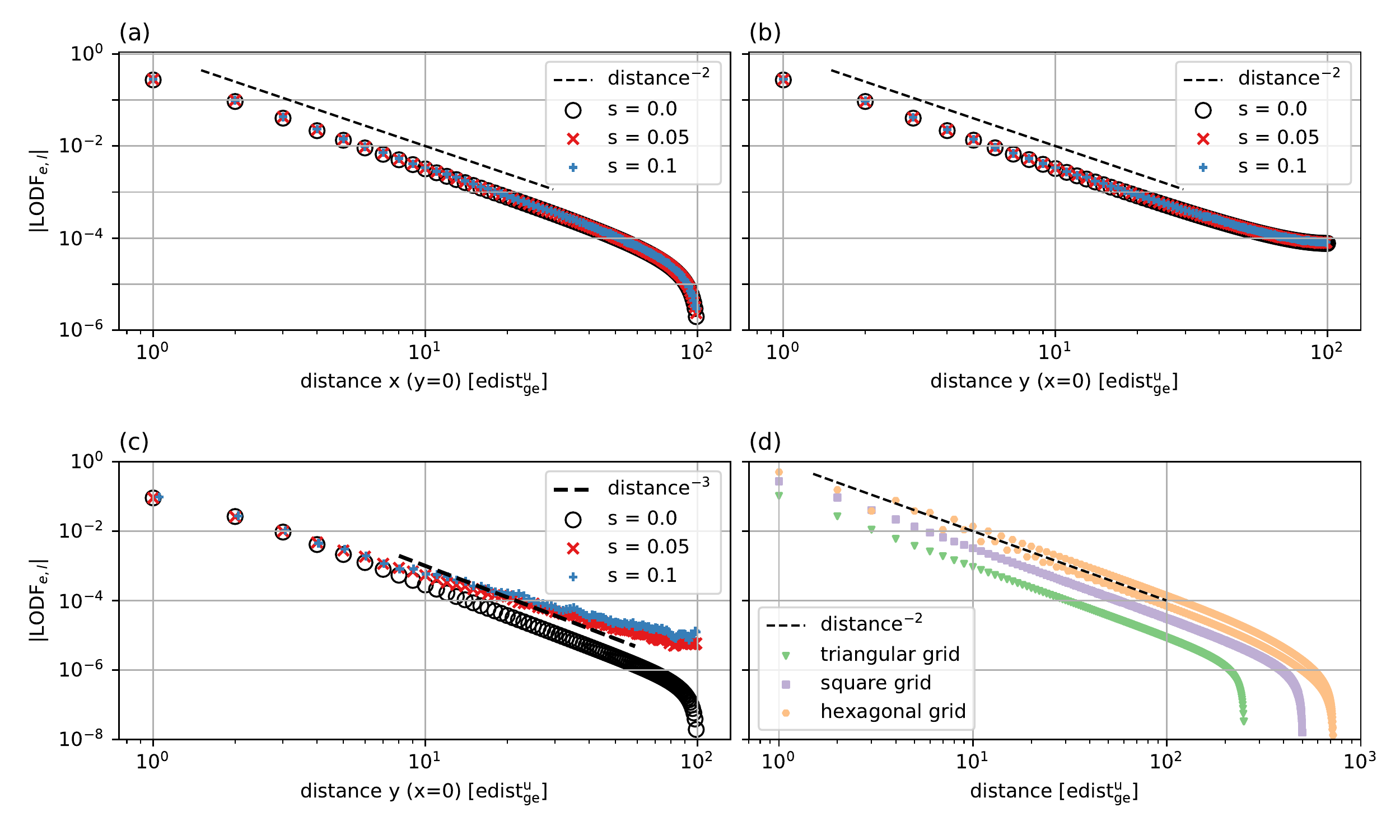}
    \end{center}
    \caption{Scaling of LODFs versus geodesic distance to failing edge for different unweighted topologies and different levels of sparsities.
    (a--c) LODFs are evaluated in different directions from the link failure and averaged over 100 realizations of square lattices from which a fraction of $s=0$ (black circles), $s=0.05$ (red crosses) and $s=0.1$ (blue plusses) links was removed randomly. %
    The failing edge is assumed to be located in $x$-direction at the center of a square grid of size $201\times 202$, cf.~figure~\ref{fig:lattice}. LODFs are calculated for (a) links along the $x$-direction (between $(x,0)$ and $(x+1,0)$), (b) links along the $y$-direction parallel to failing link (between $(0,y)$ and $(1,y)$) and (c) links along the $y$-direction perpendicular to the failing link (between $(0,y)$ and $(0,y+1)$). The $\rm dist^{-2}$ (a,b) and $\rm dist^{-3}$ (c) scaling agrees with the dipole scaling predicted using Equation~(\ref{eq:dipole-field}) as indicated by black lines. The levels of sparsity considered here do not show any effect on the scaling when considering directions parallel to the dipole axis (a,b), but the scaling becomes more long-ranged with increasing sparsity in direction perpendicular to this axis (c).
    (d) The $\rm dist^{-2}$ scaling is not unique to square grids (purple squares, size $1000\times 1000$) but may also be observed for the two other regular tilings, namely the hexagonal grid (orange hexagons, $150\times 150$ hexagons) and the triangular grid (green triangles, size $1001\times 500$). LODFs were again computed along the shortest path in $x$-direction for links oriented parallel to the dipole. The branching for the hexagonal grid is due to the fact that the path in $x$-direction is non-unique and non-straight here, such that one of the shortest paths was chosen arbitrarily. Deviations from the scaling occur for large distances due to finite-size effects.
    }
    \label{fig:sparse-lattice-dist}
\end{figure*}

We can now formally divide left-hand side (\ref{eq:PLHS-cont}) and right-hand side
(\ref{eq:PRHS-cont}) by $h^2$ and take the limit $h\rightarrow 0$ to obtain the final
continuum limit of the Poisson equation,
\be
    \nabla \cdot \left(b(x,y) \nabla \psi \right) =  - \vec q^\top \nabla \delta(x-x_r,y-y_r),
\ee
where the source term is $\vec q(x_r,y_r) = \vec F(x_r,y_r)$,
the unperturbed current field.
We note that we obtain the same continuum limit regardless of whether we use
Equation (\ref{eq:Poisson-1}) or (\ref{eq:Poisson}) to do the expansions.
Thus, the non-locality that is encoded in Equation (\ref{eq:Poisson})
is lost in the continuum formulation.

If the link weights are homogeneous, $b(x,y)= b$, and the failing link is assumed to be located at the origin $(x_r,y_r)=(0,0)$ the solution is given by the
well-known two-dimensional dipole field
\begin{align}
   \psi(\vec r) &= \frac{\vec q \cdot \vec r}{\| \vec r \|^2} \\
   \Delta \vec F(\vec r) &= b\cdot\left(\frac{\vec q }{\| \vec r \|^2} - 2 \vec r \frac{\vec q \cdot \vec r}{\| \vec r \|^4}\right) .
   \label{eq:dipole-field}
\end{align}
We thus obtain a fully analytic solution in the continuum limit. This solution reveals that the impact of link failures decays algebraically in homogeneous lattices. We consider this decay along two different axes. Assume the dipole to be located at the origin in $x$-direction, such that $\vec{q}=(\varepsilon,0)^\top$ where $\varepsilon\ll 1$ is some small real number. First, consider the decay in $x$-direction where $\vec{r}=(x,0)^\top$. In this case, we obtain for the decay of the potential and the flow changes
\begin{align*}
     \psi((x,0)^\top)&=\frac{\varepsilon\cdot x}{x^2}\varpropto\frac{1}{x}\\
     \Delta \vec F((x,0)^\top)&=b\left(\frac{\varepsilon}{x^2}-2x\frac{\varepsilon x}{x^4},0\right)^\top=-b\left(\frac{\varepsilon}{x^2},0\right)^\top.
\end{align*}
This decay in the flow changes may also be observed in the discrete version of Equation~(\ref{eq:PRHS-cont}) and is shown in Figure~\ref{fig:sparse-lattice-dist}, (a), for a line failure in a discrete square grid. Along the same lines, we may quantify the decay in $y$-direction $\vec{r}=(\varepsilon,y)^\top$ for the same dipole orientation. In this case, we obtain 
\begin{align*}
     \psi((\varepsilon,y)^\top)&=\frac{\varepsilon^2}{y^2}\varpropto\frac{1}{y^2}\\
     \Delta \vec F((\varepsilon,y)^\top)&\approx b\left(\frac{\varepsilon}{y^2},-2\frac{\varepsilon^2}{|y|^3}\right)^\top.
\end{align*}
Here, we assumed the position vector to be dominated by its $y$-component $\varepsilon\ll y$ such that $||\vec r||\approx |y|$. In total, we observe a $y^{-3}$-scaling in the flow changes in $y$-direction perpendicular to the dipole source and a $y^{-2}$-scaling in $y$-direction parallel to the dipole source, see Figure~\ref{fig:sparse-lattice-dist}, (a--c).

\section{Rigorous bounds on the dipole strength}
\label{sec:dipole}

\begin{figure}[tb]
\begin{center}
\includegraphics[width=\columnwidth]{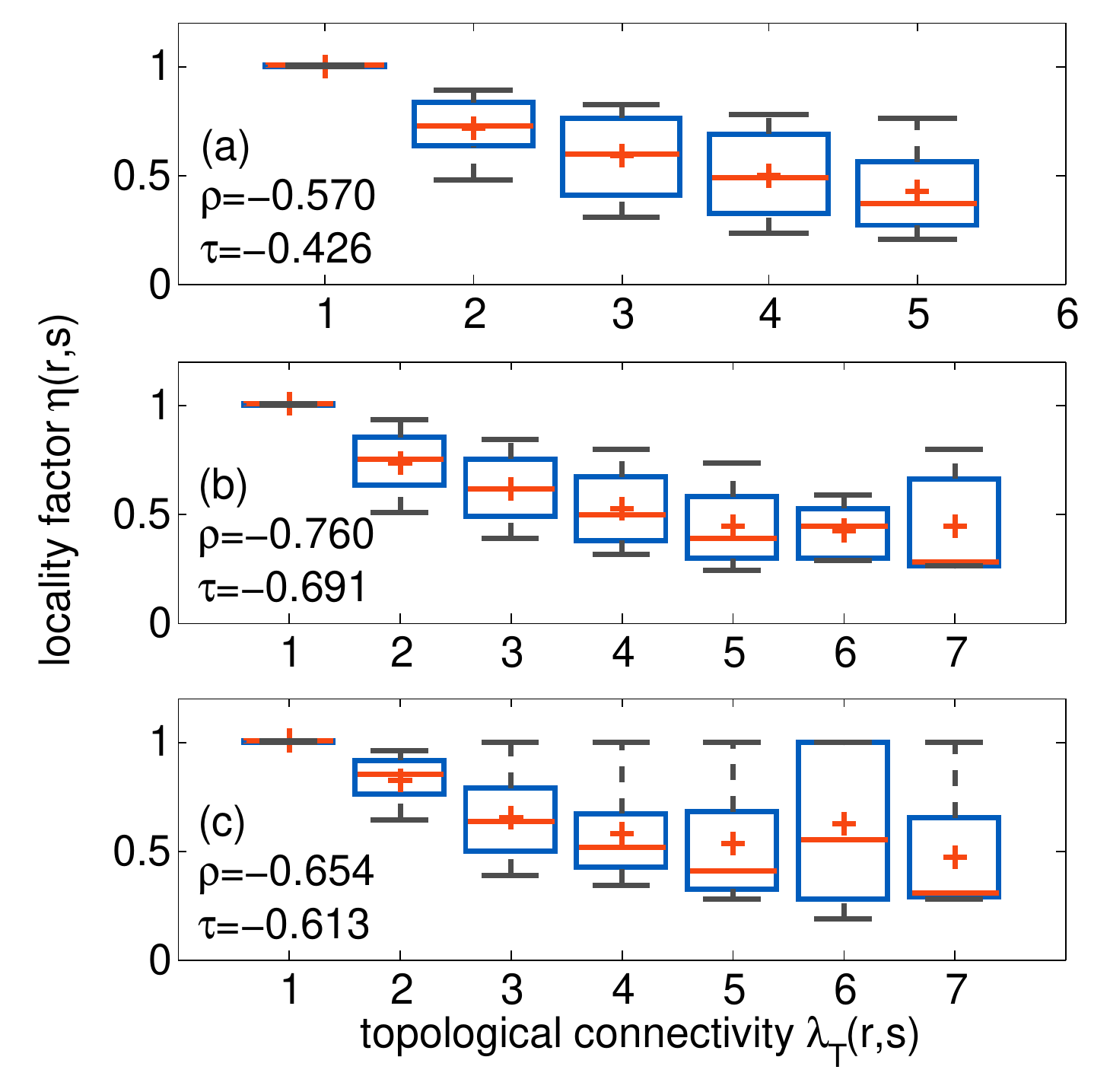}
\end{center}
\caption{
\label{fig:eta-topcon}
The locality factor $\eta(r,s)$ generally decreases with the topological connectivity $\lambda_T(r,s)$.
Values of $\eta(r,s)$ for all links $(r,s)$ with given value of $\lambda_T(r,s)$
are shown in a box-whisker-plot: the cross gives the mean, the read line the
median, the box the 25\%/75\% quantiles and the and the grey horizontal line the
9\%/91\% quantiles.
Results are shown for three standard test grids:
(a) `case118',
(b) `case1354pegase',
(c) `case2383wp'~\cite{matpower}.
The values of Pearson's correlation coefficient $\rho$ and Kendall's rank correlation coefficient
$\tau$ are given for each test grid.
}
\end{figure}

\begin{figure}[tb]
\begin{center}
\includegraphics[width=\columnwidth]{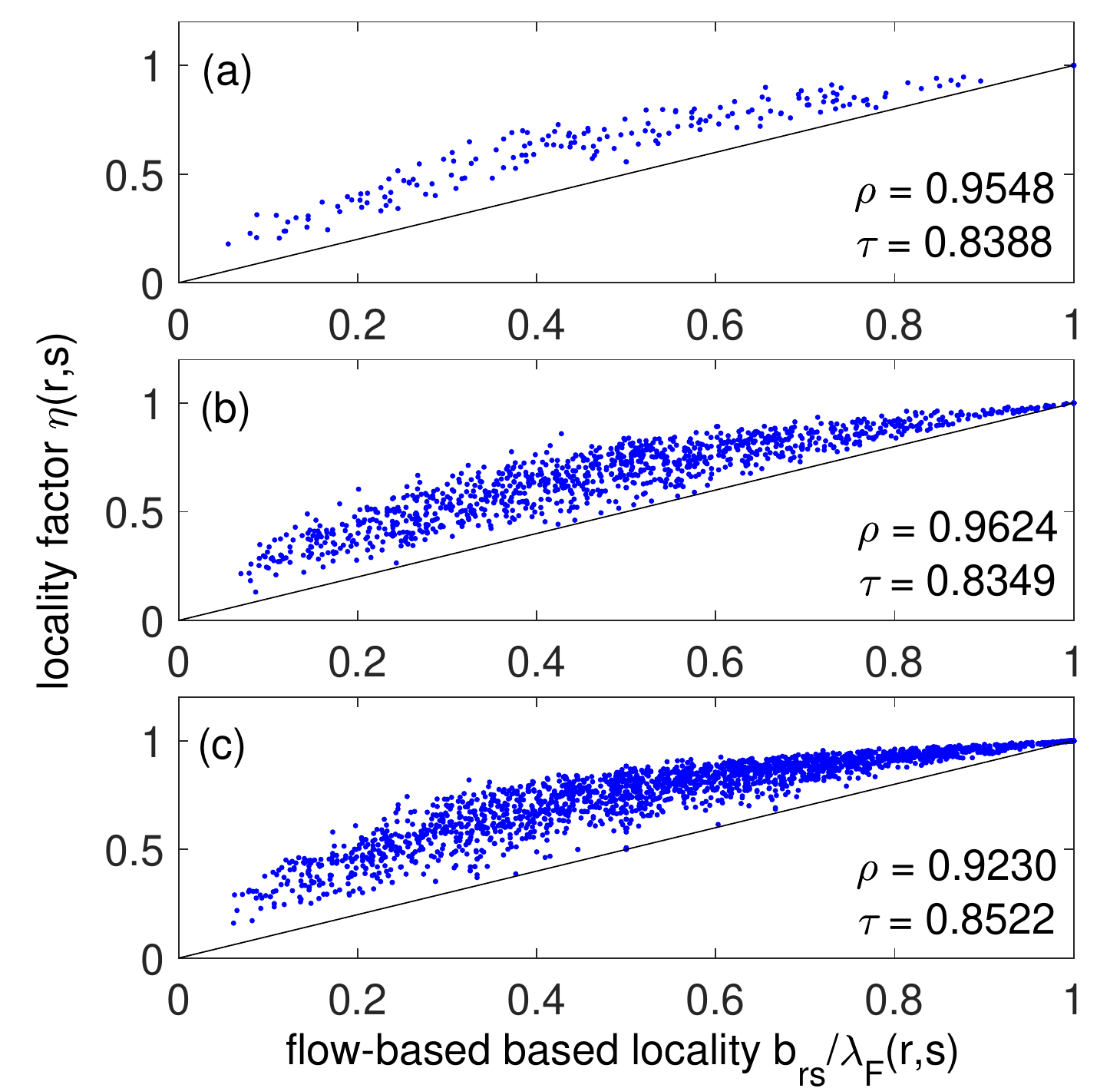}
\end{center}
\caption{
\label{fig:eta-cutcon}
The locality factor $\eta(r,s)$ is estimated by the topology based measure
$b_{rs}/\lambda_F(r,s)$ with high quality.
Results are shown for three standard test grids:
(a) `case118',
(b) `case1354pegase',
(c) `case2383wp'~\cite{matpower}.
The values of Pearson's correlation coefficient $\rho$ and Kendall's rank correlation coefficient
$\tau$ are given for each test grid. The black line is the lower bound given by
proposition \ref{prop:eta-bounds}.
}
\end{figure}

\begin{figure}[tb]
\begin{center}
\includegraphics[width=\columnwidth]{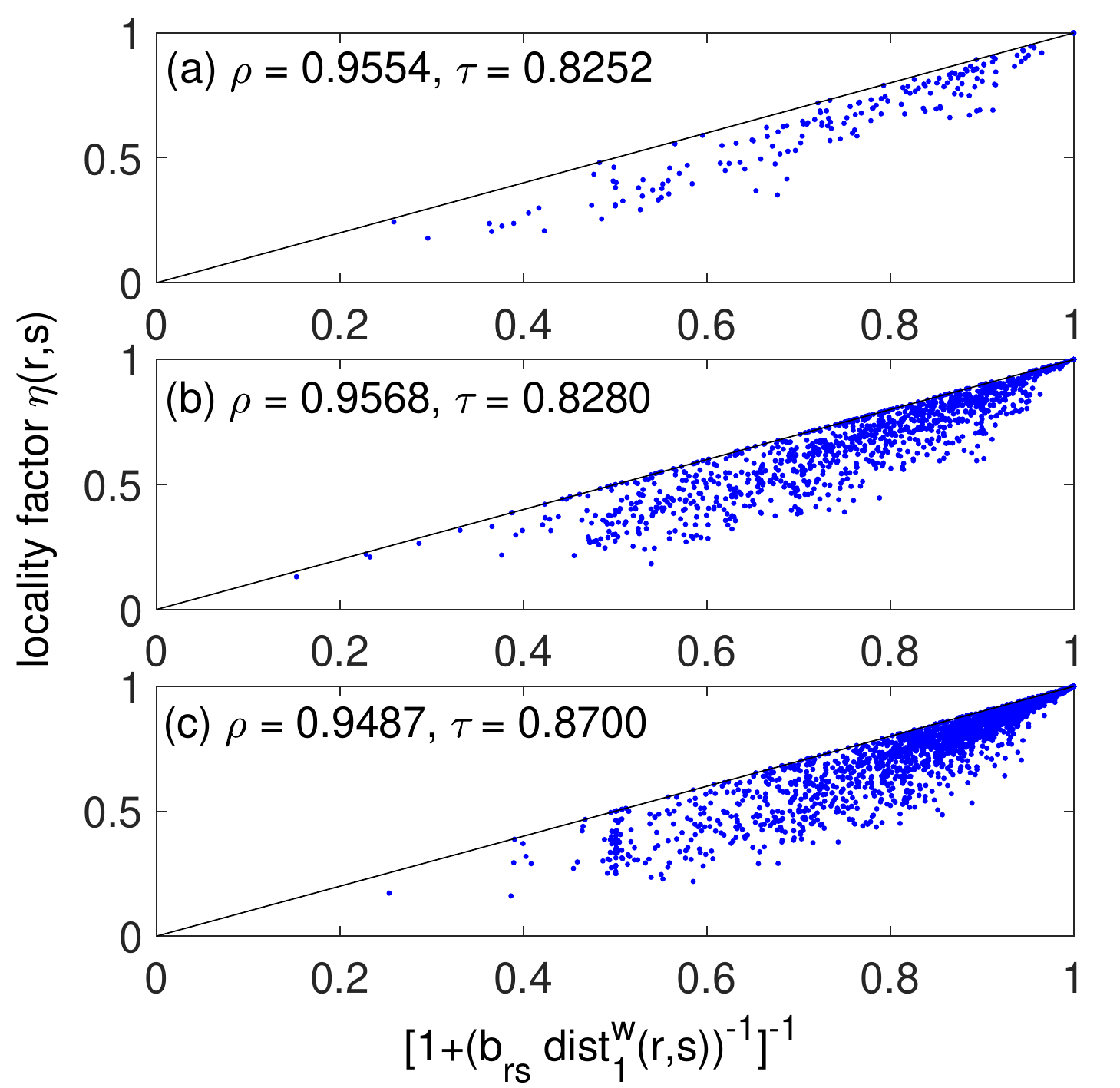}
\end{center}
\caption{
\label{fig:eta-shortpath}
An upper bound for the locality factor $\eta(r,s)$ is found in terms of the length of the shortest alternative path from $r$ to $s$, assigning to each link $(m,n)$ a weight $b_{mn}^{-1}$. The black line is the lower bound given by proposition \ref{prop:eta-upperbound} and the blue dots give results for all links in three standard test grids:
(a) `case118',
(b) `case1354pegase',
(c) `case2383wp'~\cite{matpower}.
High values of Pearson's correlation coefficient $\rho$ and Kendall's rank correlation coefficient $\tau$ show that the expression in Proposition~\ref{prop:eta-upperbound} provides a good estimate for the locality factor $\eta(r,s)$, not only a lower bound.
}
\end{figure}

We now turn to realistic networks with irregular topologies. The change in the nodal potentials or voltage phase angles $\psi_n$ and flows $\Delta F_{m \rightarrow n}$ is determined by the discrete Poisson equation (\ref{eq:Poisson}). We first consider the right-hand side of this equation, the dipole strength, which describes the gross response of the network flows to the outage. This response is proportional to the initial flow of the failing edge $F_{rs}$ and the factor
\be
    (1 - b_{rs} \vec \nu_{rs}^\top \vec B^\dagger  \vec \nu_{rs})^{-1} =: (1 - \eta(r,s))^{-1}
\ee

The factor $1-\eta(r,s)$ describes the \emph{nonlocality} of the network response to a local perturbation at link $(r,s)$. To see this, consider a grid where the real power $\Delta P$ is injected and withdrawn at the terminal nodes of the link $(r,s)$. The direct flow over the link is given by
\be
   F_{r \rightarrow s} = b_{rs} \vec \nu_{rs}^\top \vec B^{\dagger} \vec \nu_{rs} \,  \Delta P
       = \eta(r,s) \Delta P
\ee
whereas the total flow is just given by $\Delta P$. The factor $\eta(r,s)$ thus measures the fraction of the flow which is transmitted  directly and $1-\eta(r,s)$ is the fraction transmitted non-locally via other pathways. Hence, $1-\eta(r,s)$ can also be seen as a measure of redundancy. A high non-local flow indicates that there are strong alternative routes from $r$ to $s$ in addition to the direct link $(r,s)$. If no alternative path exists, the flow must be routed completely via the direct link such that $1-\eta(r,s) = 0$.

We conclude that the properties of alternative and direct paths is decisive for the understanding of flow rerouting. Before we proceed, we thus review the formal definition of a path in graph theory.
\begin{defn}
\label{def:path}
A  path from vertex $r$ to vertex $s$ is defined as an ordered set of vertices
\be
     (v_0 = r, v_1, v_2, \ldots, v_k = s),\nonumber
\ee
where two subsequent vertices must be connected by an edge and no vertex is visited twice.
Two paths are called independent if they share no common edge. The unweighted length of such a path is defined as the number of steps $k$, while the weighted path length is given by the sum of the edge weights along the path, $\sum_{j=1}^k w_{v_{j-1} v_j}$. In this work, the edge weights are given by the inverse susceptances $w_{ij} = 1/b_{ij}$. The geodesic or shortest path distance of two vertices $r$ and $s$ is defined as the length of the shortest path from $r$ to $s$.
\end{defn}

The interpretation as a redundancy measure directly relates the factor $1-\eta(r,s)$ to the topology of the network. A first rough estimate can be obtained from the topological connectivity $\lambda_T(r,s)$, which is defined as the number of independent paths from node $r$ to node $s$. A comparison for several test grids in Figure \ref{fig:eta-topcon} shows that $\eta(r,s)$ decreases with $\lambda_T(r,s)$ on average as expected, but that there is a large heterogeneity between the links.

To obtain a better topological estimate for the locality factor we need to take into account the heterogeneity of the link weights. The topological connectivity $\lambda_T(r,s)$ counts the minimum number of edges which have to be removed to disconnect the nodes $r$ and $s$. We can define a weighted analog $\lambda_F(r,s)$ as the \emph{minimum capacity} which has to be removed to disconnect the nodes $r$ and $s$. This is a classical problem in graph theory, where it is referred to as the \emph{minimum cut}~\cite{Dies10}.
We will now elaborate this quantity in a definition. An \emph{$(r,s)$-cut} can be defined as follows. Let $r\in S\subset V$ and $s\in V\setminus S$ be two vertices taken from the two disjoint sets. The $(r,s)$-cut is defined as the set of edges $\delta(S) = \{(u,v)\in E\ |\ u\in S,\ v\in V\setminus S \text{ or } v\in S,\ u\in V\setminus S \}$ connecting the two disjoint vertex-sets. The set of edges $\delta^+(S) = \{(u,v)\in E\ |\ u\in S,\ v\in V\setminus S \}$ is referred to as the \textit{forward edges} of the cut. The capacity $C$ of a cut $\delta(S)$ and the corresponding minimum capacity $\lambda_F(r,s)$ between $r$ and $s$ are then given by
\begin{align*}
     C(\delta(S)) &= \sum_{(i,j)\in \delta^+(S)} b_{ij},\\
     \lambda_F(r,s) &= \min_{\{S\subset V\ |\ r\in S,\ s\in V\setminus S\}} C(\delta(S)).
\end{align*}
By virtue of the max-flow-min-cut theorem \cite{Ahuj13}, $\lambda_F(r,s)$ is equivalent to the maximum flow which can be transmitted from $r$ to $s$ respecting link capacity limits:
\begin{align}
    \label{eq:maxflow}
    & \lambda_F(r,s) = \max_{\vec F} \sum_{n=1}^N F_{r \rightarrow n} \nn \\
    & \mbox{such that} \;
     |F_{mn}| \le b_{mn}  \; \forall \, \mbox{edges} \, (m,n) \nn \\
    & \qquad  \qquad \sum_{n=1}^N F_{mn} = 0 \; \forall \, m \neq r,s
\end{align}
Numerous efficient algorithms exist to calculate this maximum flow without performing the optimization explicitly \cite{Ahuj13}. The ratio $b_{rs}/\lambda_F(r,s)$ then gives the ratio of direct flow to total flow from $r$ to $s$ and thus provides an adequate topology-based estimate for the locality factor $\eta(r,s)$. Indeed, we can prove that it provides a rigorous lower bound.
\begin{prop}
\label{prop:eta-bounds}
The algebraic locality factor $\eta(r,s)$ is bounded by
\be
    \frac{b_{rs}}{\lambda_F(r,s)} \le \eta(r,s) \le 1.\nonumber
\ee
\end{prop}
A proof is given in appendix \ref{sec:proof-eta-bounds}. Numerical simulations for several test grids reported in Figure \ref{fig:eta-cutcon} reveal that the topological estimate not only provides a lower bound, but a high-quality estimate for the algebraic locality factor. The Pearson correlation coefficient $\rho$ between $\eta(r,s)$ and $b_{rs}/\lambda_F(r,s)$ exceeds $0.92$ for the three grids under consideration.

We arrive at the %
conclusion that %
the dipole strength given by $F_{rs} (1- \eta(r,s))^{-1}$ generally decreases with the redundancy measures $\lambda_T(r,s)$ and $\lambda_F(r,s)$.

An upper limit for the locality factor $\eta(r,s)$ can be obtained from an elementary topological distance measure. We consider the weighted geodesic distance of the two nodes $r$ and $s$ \emph{after} the failure of the direct link $(r,s)$, which we denote by ${\rm dist}_1^\mathrm{w}(r,s)$. The superscript w stands for weighted distance, the subscript 1 for the distance measured in the graph after removal of the link $(r,s)$.
We then have the following upper bound.
\begin{prop}
\label{prop:eta-upperbound}
The algebraic locality factor $\eta(r,s)$ is bounded from above by
\be
     \eta(r,s) \le  \left[ 1 + \frac{1}{b_{rs} \times {\rm dist}_1^\mathrm{w}(r,s)} \right]^{-1} \, .
     \nonumber%
\ee
\end{prop}
A proof is given in appendix \ref{sec:proof-eta-upperbound}. Numerical simulations for several test grids reported in Figure \ref{fig:eta-shortpath} reveal that the estimate in terms of the shortest path length not only provides an upper bound, but a high quality estimate for the algebraic locality factor. The Pearson correlation coefficient $\rho$ exceeds $0.94$ for the three grids under consideration.

We further note that the factor $\vec \nu_{rs}^\top \vec B^{\dagger} \vec \nu_{rs}$  can also be interpreted as a distance measure -- the resistance distance \cite{Klei93,Xiao2003}. We come back to the quantification of distances in flow networks later in section \ref{sec:decay-num}.

\section{Spatial distribution of flow rerouting}

We now turn to the spatial aspects of flow rerouting in general network topologies. We first discuss some rigorous results, showing how the network topology determines the rerouting flows. Then, we return to the regular tilings and study the effect of increasing sparsity in these topologies on the dipole scaling. Finally, we suggest a new measure of distance for flow rerouting and examine its performance on realistic network topologies taken from power grids. %

\subsection{Rigorous results}
\label{sec:rigorous}

To start off, we first present a lemma due to Shapiro~\cite{Shap87}, relating the flow changes after a link failure in an unweighted graph solely to the topology of the underlying network.
\begin{lemma}
Consider an unweighted network with a unit dipole source along the edge $(r,s)$, i.e. a unit inflow at node $r$ and unit outflow at node $s$. Then the flow along any other edge $(m,n)$ is given by
\begin{equation}
    F_{m \rightarrow n} = \frac{
     \mathcal{N}(r,m\rightarrow n,s)-\mathcal{N}(r,n\rightarrow m,s)}{\mathcal{N}},\nonumber
\end{equation}
where $\mathcal{N}(r,m\rightarrow n,s)$ is the number of spanning trees that contain a path from $r$ to $s$ of the form $r,\ldots,m,n,\ldots,s$ and $\mathcal{N}$ is the total number of spanning trees of the graph. \label{lem:electricallemma}
\end{lemma}
This lemma exactly gives the LODFs in terms of purely topological properties -- the number of spanning trees containing certain paths. A generalization of this theorem to weighted graphs was recently presented in~\cite{Guo17}.

However, counting spanning trees is typically a difficult task such that these results are of limited use for practical applications. Nevertheless, they reveal the importance of certain paths through networks which we will analyze numerically in more detail below. Before we turn to this issue, we derive some weaker, but more easily applicable rigorous results. 

We expect that the flow changes $\Delta F_{mn}$ decay with distance as for the case of the square lattice analyzed in section \ref{sec:square}. Can we establish some rigorous results on the decay with distance for arbitrary networks? Consider the outage of a single edge  and assume that the network remains connected afterwards. We label the failing link as $(r,s)$ such that $F_{r \rightarrow s}> 0$ w.l.o.g. We first consider the change of the nodal potential or voltage phase angles $\psi_n$ and its decay with distance to the failing link $(r,s)$. More specifically, we define the maximum and minimum values of $\psi_n$ attained at a given distance:
\begin{align} 
   u_d   &= \max_{n, {\rm dist}_0^\mathrm{u}(n,r) = d}  \psi_n \nn \\
  \ell_d &=  \min_{n, {\rm dist}_0^\mathrm{u}(n,s) = d}  \psi_n.\nonumber
\end{align}
Here, $\text{dist}_0^\mathrm{u}(n,r)$ denotes the geodesic distance between two nodes $n$ and $r$ in the initial unweighted graph (indicated by the superscript $u$ for unweighted and subscript $0$ for the initial pre-contingency network). We then find the following result.

\begin{prop}
\label{prop:decay}
Consider the failure of a single link $(r,s)$ with $F_{r \rightarrow s} >0$ in a flow network.
Then the maximum (minimum) value of the potential change $\psi_n$ decreases
(increases) monotonically with the distance to nodes $r$ and $s$, respectively:
\begin{align}
   & u_d  \le u_{d-1}, \qquad 1 \le d \le d_{\rm max}.  \nn \\
   & \ell_d  \ge \ell_{d-1},  \qquad 1 \le d \le d_{\rm max}.\nonumber
\end{align}
\end{prop}
A proof is given in appendix \ref{sec:proof-decay}.
We thus find that potential changes generally decrease with the distance
in magnitude and so do the flow changes.

Furthermore, we can exploit the analogy to electrostatics to gain an insight into the scaling of flow changes with distance. As the flows are determine by a discrete Poisson equation, a discrete version of Gauss' theorem follows immediately. We note that we formulate this result in terms of the original network topology, cf.~Equation (\ref{eq:Poisson}).

\begin{lemma}
\label{lem:Gauss}
Consider the failure of a single link $(r,s)$ in a flow network and denote by $V$ be the set of
vertices in the network.
For every decomposition of the network $V = V_1 + V_2$ with $r \in V_1$ and
$s \in V_2$ we have
\be
   \sum_{m \in V_1, n \in V_2} F_{m \rightarrow n} = F_{rs} \, (1-\eta(r,s))^{-1} \,\nonumber .
\ee
That is, for each decomposition the total flow between the two parts $V_1$ and $V_2$ equals the dipole strength.
\end{lemma}
For general network topologies this lemma implies that the average flow will decay with distance from the failing link $(r,s)$: Choose $V_1$ to include all nodes which are closer to $r$ than to $s$ and have a distance to $r$ smaller than a given value
$$
   V_1 = \{ n \in V | \mbox{dist}_0^\mathrm{u}(r,n) \le d \; ; \mbox{dist}_0^\mathrm{u}(r,n) \le \mbox{dist}_0^\mathrm{u}(s,n) \}.
$$
With increasing value of $d$ the number of nodes in $V_1$ increases and typically the number of edges between $V_1$ and $V_2$ increases, too. The total flow over these links remains constant according to lemma \ref{lem:Gauss}, such that the average flow will generally decrease. The exact scaling of the number of edge between $V_1$ and $V_2$ of course depends on the topology of the network.

\begin{figure*}[tb]
\begin{center}
\includegraphics[width=\textwidth]{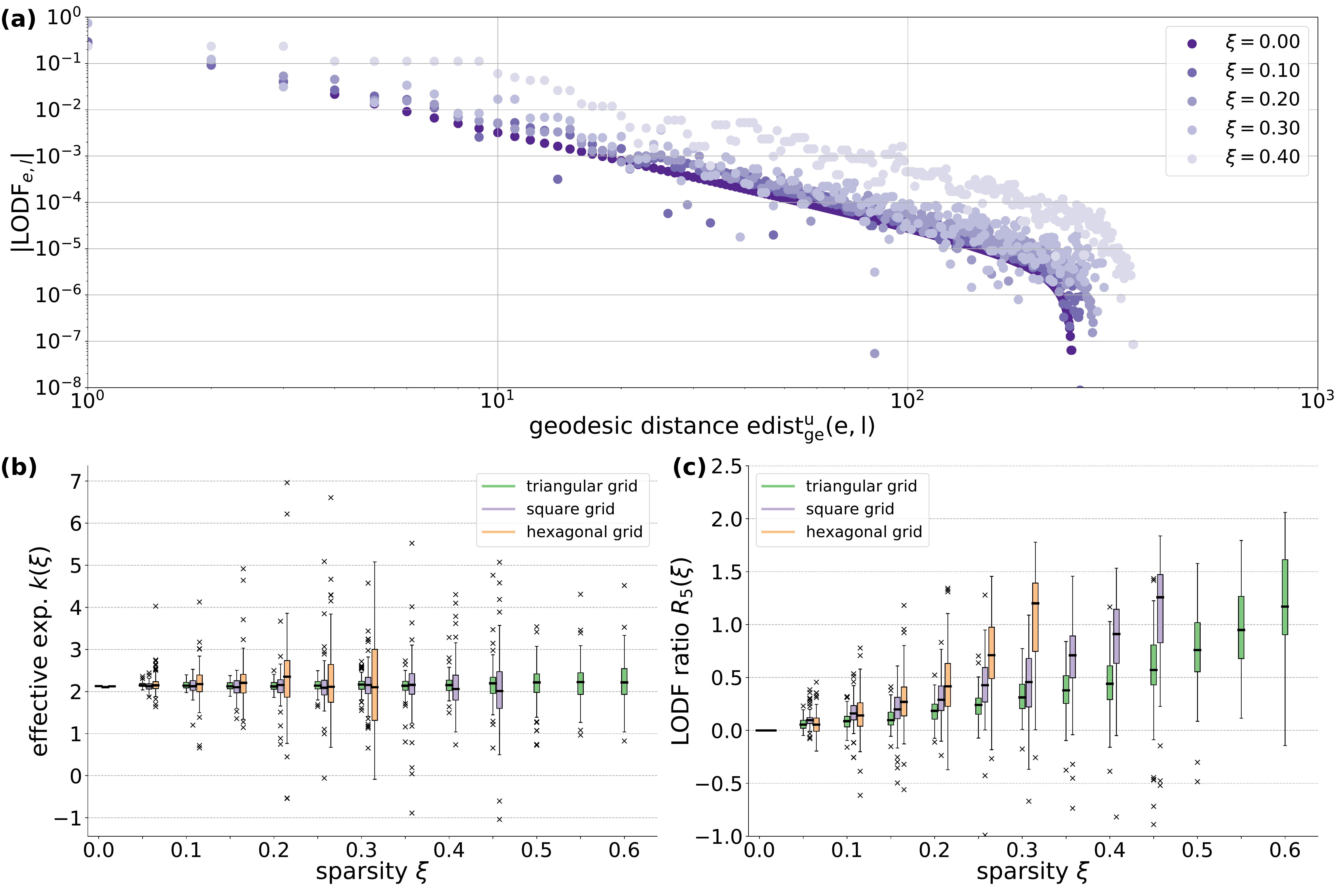}
\end{center}
\caption{ \label{fig:bp-exponents} Increasing sparsity leads to more long-ranged effects of link failures in regular grids.(a) Exemplary scaling of LODFs in a square grid of size $500\times 500$ with increasing sparsity (colors from dark to light purple), now achieved through the removal of edges not contained in an arbitrary spanning tree. (b) Whereas the effective exponent shows no change and thus still obeys approximately the inverse-square law for all topologies, (c) the logarithmic ratio between LODFs with and without sparsity at a certain distance increases on average with increasing sparsitiy. Boxplots are shown for $100$ realizations of the hexagonal grid (left, orange), square grid (purple, center) and triangular grid (green, left) choosing a random spanning tree as the basis for edge removal for each realization and value of sparsity.}
\end{figure*} 

One can furthermore show that a sufficient connectivity is needed for perturbations to spread. Generally, flow can be rerouted via an edge $(m,n)$ only if it can enter and leave the link via two independent paths. One can thus prove the following statement \cite{17lodf,Guo17}.

\begin{prop}
\label{prop:decoupling}
The line outage distribution factor $\mbox{LODF}_{e,\ell}$ between two edges $e = (m,n)$ and $\ell = (r,s)$ vanishes if there are less than two independent paths between the vertex sets $\{r,s\}$ and $\{m,n\}$.
\end{prop}

\subsection{Impact of network topology}

 Now that we derived rigorous results on the scaling of LODFs, we want to study the influence of network connectivity on the scaling in more detail.
 
 To do so, we first compare the scaling obtained for the square grid to the one in the other two regular tilings of two-dimensional space, namely the hexagonal grid and the triangular grid. In perfect realizations of these grids, each node has a degree of $\operatorname{deg}_{\text{hex}}=3$ and $\operatorname{deg}_{\text{tri}}=6$, respectively, whereas the degree for the square grid reads $\operatorname{deg}_{\text{sg}}=4$. In Figure~\ref{fig:sparse-lattice-dist}(d), the LODFs are evaluated for these three topologies with increasing geodesic distance from the failing edge located again in the center of the networks between the nodes at $(x_r,y_r)=(0,0)$ and $(x_s,y_s)=(1,0)$. The quadratic scaling with the geodesic distance in $x$-direction $\|x\|^{-2}$ (black, dotted line) is preserved for all three topologies, i.e. the triangular grid (green triangles, bottom), the square grid (red squares, center) and the hexagonal grid (blue hexagons, top). The grids used here were of size $1000\times 1000$ and $1001\times 500$ nodes for the square grid and the triangular grid, respectively, and $150\times 150$ hexagons for the hexagonal grid.
 
 Thus, the quadratic scaling is robust throughout different regular networks. However, real networks are in general not regular. For this reason, we proceed by studying the effect of increasing sparsity in these regular tilings. Define the sparsity $\xi \in[0,1]\subset \mathbb{R}$ as the \textit{fraction of edges removed from the original graph}. We make use of two different methods to achieve increasing sparsity. Our first method is a completely random removal of edges in the graph followed by measuring the LODFs along a specified path. If an edge along the path does no longer exist, we simply skip the edge. The results obtained from this method are shown in Figure~\ref{fig:sparse-lattice-dist} (a--c). There is no change visible in the scaling of LODFs, except for the direction perpendicular to the dipole in panel (c). In particular, only small values of sparsity $\xi$ can be studied using this method, since a random removal of edges may easily result in disconnected graphs. For this reason, we make use of another method. %
 
 For the second method, we first construct an arbitrary spanning tree of the network after removal of the failing edge. Then, we subsequently remove random edges from the graph that are not part of the tree until a fraction $\xi$ of its original edges is removed from the graph. This way, we make sure that the whole graph stays connected at all times. We continue by constructing the shortest path from the failing edge $((0,0),(1,0))$ to the node located at $(x_{\rm max},0)$ and quantify the LODFs along this path. Note that using this method to make a graph sparser, we need to take into account the graph-specific maximal sparsity $\xi_{\text{max},G}$, i.e. the fraction of edges whose removal would disconnect the graph. Assuming the initial tree to be minimal, this fraction may be calculated as $\xi_{\text{max,hex}}=1/3$, $\xi_{\text{max,sg}}=1/2$ and  $\xi_{\text{max,tri}}=2/3$ for the hexagonal grid, square grid and triangular grid, respectively. 
 
 Using this procedure, we can quantify the scaling of LODFs in grids with increasing sparsity. The direct assessment of a scaling exponent is difficult for sparser graphs due to the large spread in LODF values, see Figure~\ref{fig:bp-exponents}(a). This is why we construct a different measure to quantify this scaling. We consider the \textit{effective exponent} $k(\xi)$, where $\xi$ is the graph's sparsity, and assume a scaling of the form 
 $$|\text{LODF}(r,\xi)|\varpropto r^{-k(\xi)}$$ 
 in some region of the geodesic distance $r=\| \vec r \|$ from the link failure. This effective exponent is calculated as follows 
 \begin{align*}
     k(\xi)=-\log_{5}\left(\frac{\sum_{r\in [5 \cdot 10^1-w,5 \cdot 10^1+w]}|\text{LODF}(r,\xi)|}{\sum_{r\in [10^1-w, 10^1+w]}|\text{LODF}(r,\xi)|}\right),
 \end{align*}
where $w\in\mathbb{N}$ is a window specifying the range to average over in order to smooth the LODF values considered. We chose a windows size of $w=2$ when calculating the effective exponent in practice which we found to result in a good compromise between smoothing and completely removing the trend. However, we did not observe a strong effect of the window size on the results. For a perfect inverse square law $|\text{LODF}|\varpropto \| \vec r \|^{-2}$ and a vanishing window $w=0$, this parameter yields $k=-\log_{5}(5^{-2})=2$ as required. In Figure~\ref{fig:bp-exponents}(b), it can be observed that this effective exponent stays approximately constant at $k\approx 2$ over different values of sparsity and the three different topologies considered, where results for each value of sparsity were obtained using $100$ random realizations of edge removals and with the same grid sizes as stated previously.

To further quantify the effect of increasing sparsity in regular networks, we make use of another measure which we refer to as the \textit{LODF ratio} $R_w(\xi)$. It is simply calculated as the logarithmic ratio between the LODFs with and without sparsity, again averaged over a fixed window of distances
 \begin{align*}
 R_w(\xi)=\log_{10}\left(\frac{\sum_{r\in [10^1-w,10^1+w]}|\text{LODF}(r,\xi)|}{\sum_{r\in [10^1-w, 10^1+w]}|\text{LODF}(r,\xi=0)|}\right).
 \end{align*}
 Note that we evaluate this parameter at a distance of $10^1$ but we found the parameter to yield similar values for all distances considered. A parameter of $R_w(\xi)=1$ then represents a tenfold increase in the LODFs as compared to the network without any edges removed. In Figure~\ref{fig:bp-exponents}(c), this parametzer is shown for the different topologies and sparsities. Here, a window size of $w=5$ was used. An increase with increasing sparsity is clearly visible. In particular, the LODFs increase on average more than tenfold close to the highest possible values of sparsity.
 
In total, we observe that the scaling exponent derived from the dipole analogy in section~\ref{sec:dipole} holds for the regular networks even when removing a large fraction of their edges. On the other hand, the LODF values at a certain distance from the failing link show an increase with increasing sparsity, such that the actual effect of a link failure can be up to tenfold stronger than for the corresponding regular grid with no links removed. Thus, the overall effect of a link failure is more long-ranged in a sparser network, although no change in the effective exponent can be observed.

\subsection{Scaling with distance}
\label{sec:decay-num}

\begin{figure}[tb]
\begin{center}
\includegraphics[width=6cm]{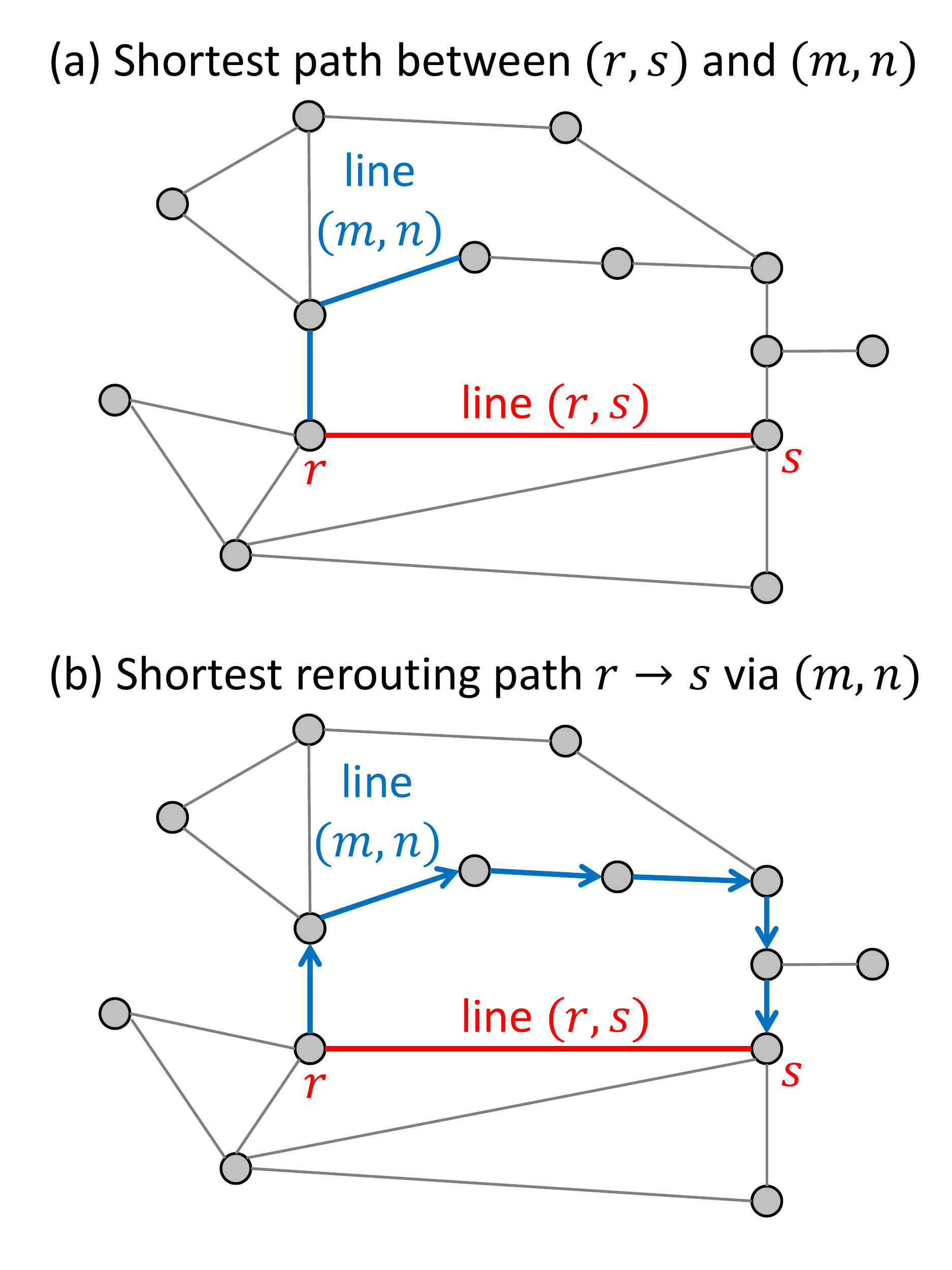}
\end{center}
\caption{
\label{fig:def_rerouting}
Illustration of two different distance between two links $(r,s)$ and $(m,n)$ (coloured in red).
(a) The common geodesic or shortest-path distance (indicated by thick lines).
(b) The rerouting distance is defined as the length of the shortest path from $r$ to $s$ crossing
the link $(m,n)$ and is indicated by thick arrows.
The sample network in this figure is based on the topology of the IEEE 14-bus test grid \cite{powertest}.
}
\end{figure}

\begin{figure*}[tb]
\begin{center}
\includegraphics[width=\textwidth]{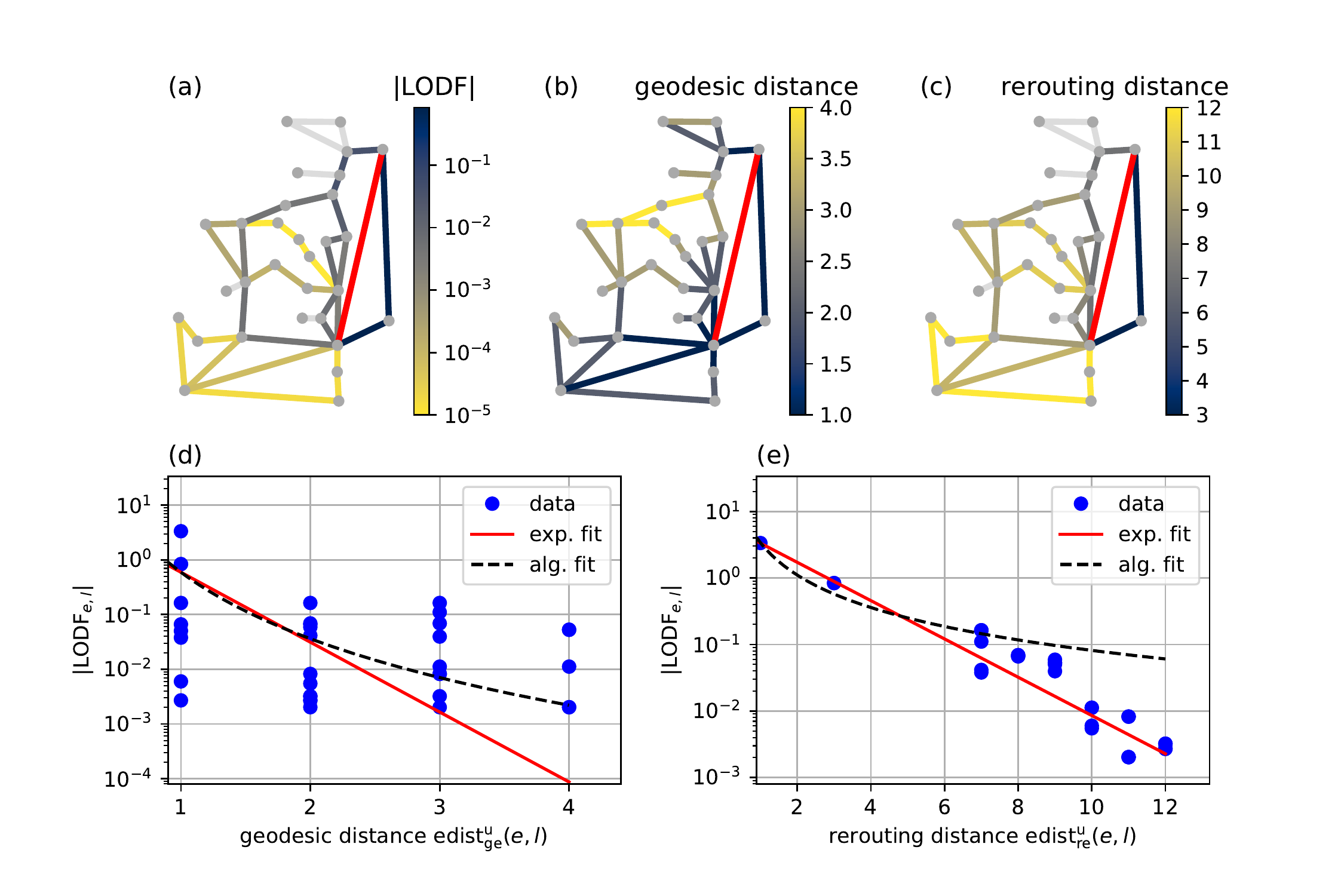}
\end{center}
\caption{
\label{fig:lodf-vs-dist-30}
Line outage distribution factors (LODFs) in comparison to the unweighted geodesic and rerouting distances.
(a) Magnitude of the LODFs in the IEEE 30-bus test grid 'case30' \cite{powertest}. The failing link $l$ is marked in red.
(b) The geodesic distance to the failing link $\text{edist}_\text{ge}^\mathrm{u}$.
(c) The rerouting distance to the failing link $\text{edist}_\text{re}^\mathrm{u}$.
(d,e) LODFs vs. geodesic and rerouting distance (blue dots) including an exponential (red solid line) and an algebraic (black dashed line) least-squares fit to the data. Due to the low number of data points, a clear decision whether the correlation is algebraic or exponential is not possible.
}
\end{figure*}

The impact of a link failure generally decays with distances. While the definition of distance is straightforward in regular lattices, different measures are meaningful in networks with complex topologies. The geodesic distance of two links follows from definition \ref{def:path} for two vertices
\begin{align*}
       &{\rm edist}^\mathrm{w}_{\text{ge}}[(r,s),(m,n)] \\
       &= \min_{v_1\in\{r,s\} , v_2 \in \{m,n\} } {\rm dist}_0^\mathrm{w}(v_1,v_2)+\frac{w_{rs}+w_{mn}}{2}. 
\end{align*}
Here, $w_{rs}$ is the edge weight assigned to the edge $(r,s)$. When considering the unweighted analog, the edge distance is defined analogously setting all edge weights to one. The additional term $\frac{w_{rs}+w_{mn}}{2}$ ensures that neighboring edges have non-zero distance, e.g. unity distance $\rm{edist}^{\rm u}_{\rm ge}=1$ in the unweighted case. However, this distance is a bad indicator for flow rerouting in real-world irregular topologies. An example shown in Figure~\ref{fig:lodf-vs-dist-30} demonstrates that this simple distance is only weakly correlated with the magnitude of the LODFs for a real-world power grid test case.

Instead, we need a distance measure based on flow rerouting. If a link $(r,s)$ fails, the flow must be rerouted through other pathways, as described by the electrical lemma \ref{lem:electricallemma}. However, it is not feasible to take into account all spanning trees which govern the flow rerouting. In order to still be able to estimate the impact on another link $(m,n)$, we will thus consider a path from $r$ to $s$ that crosses this link. The main difference to the ordinary graph theoretical distance is that we have to take into account a path \emph{back and forth}. We are thus led to the following definition.

\begin{defn}
A \emph{rerouting path} from vertex $r$ to vertex $s$ via the edge $(m,n)$ is a path
\be
    (v_0 = r, v_1, \ldots, v_i = m, v_{i+1} = n, v_{i+2}, \ldots, v_k=s)\nonumber
\ee
or
\be
    (v_0 = r, v_1, \ldots, v_i = n, v_{i+1} = m, v_{i+2}, \ldots, v_k=s)\nonumber
\ee
where no vertex is visited twice. The \emph{rerouting distance} between two edges $(r,s)$ and $(m,n)$ denoted by ${\rm edist}^\mathrm{u/w}_{\text{re}}[(r,s),(m,n)]$ is the length of the shortest rerouting path from $r$ to $s$ via $(m,n)$ plus the length of edge $(r,s)$. Equivalently, it is the length of the shortest cycle crossing both edges $(r,s)$ and $(m,n)$. If no such path exists, the rerouting distance is defined to be $\infty$. \label{def-reroute-dist}
\end{defn}
The definition of a rerouting path is illustrated in Figure \ref{fig:def_rerouting}. Again, we consider a weighted and an unweighted version of this distance indicated by the superscript $^w$ and $^u$, respectively. We note that the length of the edge $(r,s)$ is included in order to make the distance measure symmetric. In Appendix \ref{app:redistance}, we show explicitly that this definition satisfies the axioms of a metric and discuss how to compute the shortest rerouting path.

An example of rerouting distances in comparison to the LODFs is shown in Figure~\ref{fig:lodf-vs-dist-30} for a small test grid. We observe a much better correlation in comparison to the ordinary geodesic distance defined above. The limitation of geodesic distances becomes especially clear for situations described by proposition \ref{prop:decoupling}. If exactly one independent path exists between two  links, the rerouting distance is $\infty$, while the geodesic distance is finite. Hence, the latter fails to explain why the LODF between the two links vanishes.

To further investigate the importance of distance, we simulate all possible link failures in four test grids of different size. For every failing link $(r,s)$ we evaluate the geodesic distance as well the rerouting distance to all other links in the grid. To quantify to which extend the distance predicts the magnitude of the LODFs, we then calculate the Kendall rank correlation coefficient $\tau$ \cite{Kendall1948}. This coefficient is used on ordinal data and assumes values in the interval $[-1,1]$. A value of (minus) one indicates perfect (anti)correlation, whereas a zero value implies no correlation between the data. Table \ref{tab:rankcor-dist-lodf} shows the results, averaging over all trigger links $(r,s)$ in the respective grid discarding bridges. The rank correlation is negative as LODFs generally decay with  distance. The magnitude of the rank correlation is significantly higher for the rerouting distance. In particular for the test grid `case1354pegase' we see that the ordinary geodesic distance has a very limited predictive power for the LODFs ($|\tau| < 0.25$), while the rerouting distance is strongly correlated to the magnitude of the LODFs $|\tau| > 0.83$. 
Figure~\ref{fig:kendalltau_histograms} illustrates this discrepancy in the distribution of $\tau$ values for the different distance measures for the test grids 'case118' and 'case1354pegase'.
We are thus led to the conclusion that geodesic distances are of limited interest when considering the impact of link failures and should be replaced by other measures such as rerouting distances. Notably, we observe no major difference when comparing weighted and unweighted distances.

\begin{figure}[tb]
\begin{center}
\includegraphics[width=\columnwidth]{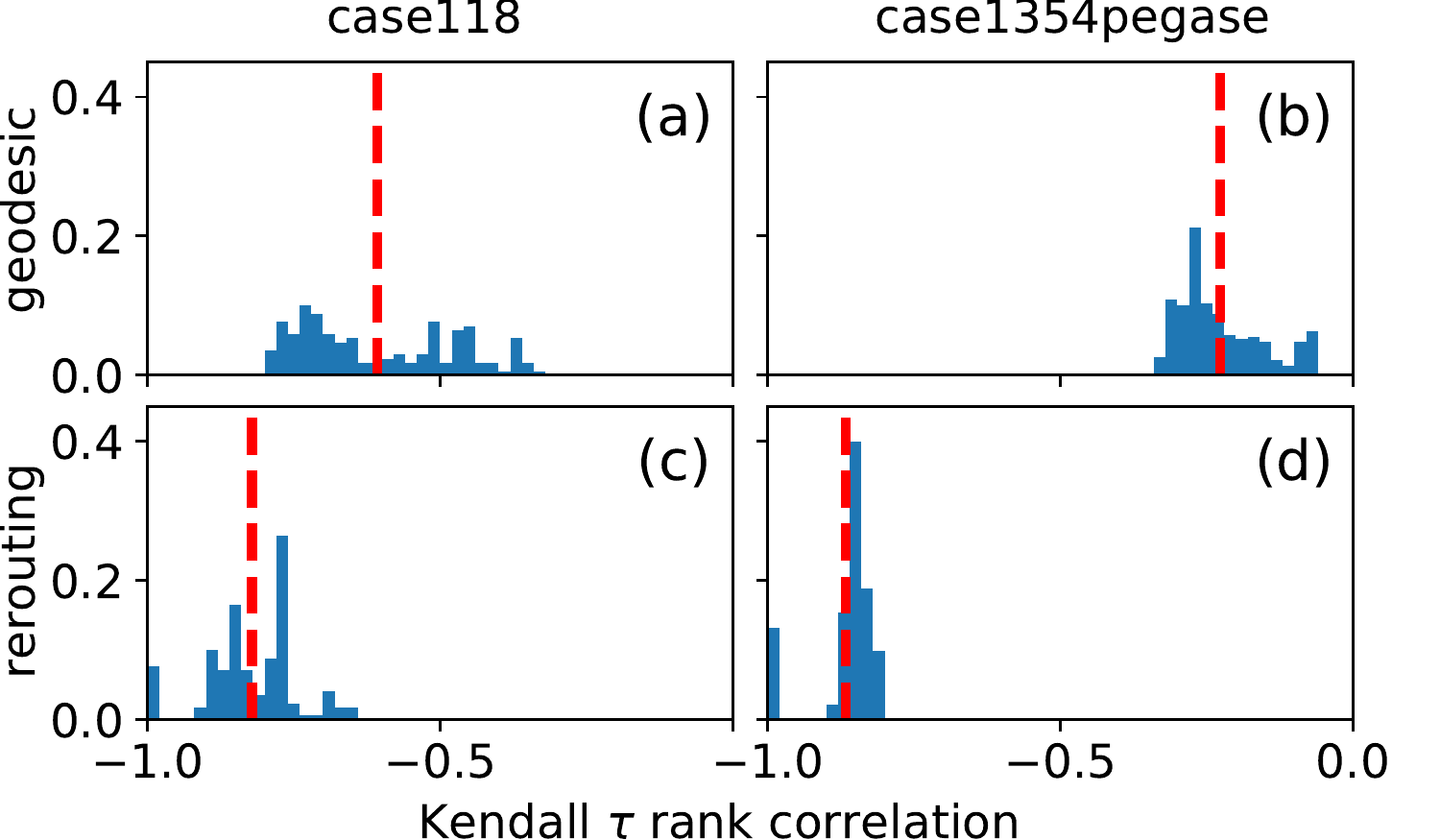}
\end{center}
\caption{
\label{fig:kendalltau_histograms}
(a,b, top) Normalized histograms of the Kendall $\tau$ rank correlation for the magnitude of LODF and unweighted geodesic distance and (c,d, bottom) $|\text{LODF}|$ and unweighted rerouting distance between two links in the IEEE test cases 118 (a,c) and 1354pegase (b,d). Vertical red lines show the average of the distribution of tau values (cf.\ table~\ref{tab:rankcor-dist-lodf}). The stronger correlation of the rerouting distance with the LODFs as compared to the geodesic distance is clearly visible.
}
\end{figure}

\section{Conclusion}
\label{sec:conclusion}

Link failures represent major threats to the operation of complex supply networks across disciplines. In this article, we examined the impact of such failures in terms of the induced flow changes, which is commonly referred to as Line Outage Distribution Factors (LODFs). We provide mathematically rigorous results and extensive numerical simulations with a focus on the gross network response (i.e. the dipole strength), the scaling of flow changes with distance and the role of network topology. These quantities are crucial to understand the global robustness of supply networks as each failure can trigger a cascade of secondary failures with potentially catastrophic consequences.  

First, we demonstrated rigorously that the flow changes created by a single failure in a square lattice corresponds to the field of an electromagnetic dipole. Hence the effects of a failure decay with the distance following an inverse-square law. The dipole analogy developed here allows for an analytical expression describing the spreading of link failures. Although this treatment is rigorously valid only in the continuum limit, we showed that the observed scaling extends to the other regular tilings of two-dimensional space even after removing a fraction of links. Thus, we conclude that the scaling may be expected to hold also for realistic topologies.

\begin{table}[htb!]
\begin{tabular}{| l | c | c | c | c |}
\hline
  test grid & \multicolumn{4}{|c|}{Rank correlation $\tau$ $|\text{LODF}|$ vs. distance} \\
\hline
    & \multicolumn{2}{|c|}{geodesic distance} & \multicolumn{2}{|c|}{rerouting distance} \\
    & unw. & weighted  & unw. & weighted  \\
\hline
  case30 & -0.4027 & -0.4015 & -0.8528 & -0.8440 \\
  case118 & -0.6069 & -0.5233 & -0.8211 & -0.7920 \\
  case1354pegase & -0.2269 & -0.1341 & -0.8664 & -0.8438 \\
  case2383wp & -0.3604 & -0.2318 & -0.7213 & -0.6066 \\
\hline
\end{tabular}
\caption{
\label{tab:rankcor-dist-lodf}
Average of the Kendall $\tau$ rank correlation values for magnitude of LODF versus different distance measures. The four different IEEE test cases consistently show a higher degree of correlation between rerouting distances and LODF than between geodesic distances and LODF in both weighted and unweighted cases, while the unweighted rerouting distance slightly outperforms the weighted one. For examplary distributions of the $\tau$ values cf.\ figure~\ref{fig:kendalltau_histograms}.
}
\end{table}

Increasing the sparsity of a network promotes more long-ranged effects up to the point when two links are only related by one independent pathway. Then, a rerouting between the two links becomes impossible and a failure of one link does not affect the other. However, this also implies a lack of redundancy such that a link failure can have catastrophic consequences locally.

In real-world irregular networks, the gross response of a failure depends on the loading of the link as well as the local network structure. Rigorous upper and lower bounds were given for the dipole strength relating it to the redundancy of the failing link. Furthermore, the common notion of a geodesic graph distance is of limited use to predict flow rerouting. We thus introduced a \textit{rerouting distance} which we showed to be much more meaningful to predict the impact of failures.

Whereas the classical analysis of link failures relies heavily on simulation results, our results provide heuristic methods and rigorous bounds which allow for an analytical insight into the relationship between the structure of a network and its robustness towards link failures. In particular for large networks where simulations are difficult, our results allow for an a priori analysis of link failures and might also be used to identify critical links, for instance in terms of the locality factor which quantifies the response of a network to a single failure. This type of analysis is aided by the general results on decay of maximal flow changes with geodesic and rerouting distances.

We expect our results to be applicable far beyond power grids since the linearized treatment extends to other phenomena such as hydraulic or biological networks. The rerouting distance along with the bounds on the locality factor may greatly simplify the study of link failures in all kinds of supply networks and makes them more accessible. We expect our results on the scaling of LODFs for networks with increasing sparsity along with this distance measure to help identifying critical parts and paths and improving the overall robustness of supply networks.

\acknowledgments

We gratefully acknowledge support from the German Federal Ministry of Education and Research (BMBF grant no. 03SF0472) and the Helmholtz Association (via the joint initiative ``Energy System 2050 -- a contribution of the research field energy'' and the grant no. VH-NG-1025 to D.W.).

\appendix

\section{Proof of Proposition \ref{prop:eta-bounds}}
\label{sec:proof-eta-bounds}

\begin{proof}
By definition, $\eta(r,s)$ is given by the flow $F_{rs}/\Delta P$ when the power $\Delta P$
is injected at node $r$ and withdrawn at node $s$, while there is no injection at any other node,
\begin{align}
    & \sum \nolimits_n F_{r \rightarrow n} =  \sum \nolimits_n F_{n \rightarrow s} = \Delta P  \nn \\
    & \sum \nolimits_n F_{m \rightarrow n} = 0 \qquad \forall \, m \neq r,s,
\end{align}
such that
\be
    \eta(r,s) = \frac{F_{r \rightarrow s}}{ \sum \nolimits_n F_{r \rightarrow n}} .
    \label{eq:eta-F-direct-total}
\ee
Using the basic relation $F_{r \rightarrow s} = b_{rs} (\theta_r - \theta_s)$ we can thus express the
inverse of $\eta(r,s)$ as
\be
   \frac{b_{rs}}{\eta(r,s)} =  \sum \nolimits_n F_{rn}
   \quad \mbox{if} \quad
   \theta_r - \theta_s = F_{rs}/b_{rs} = 1.
\ee
We can now use that the potential drop over all other links in the network is smaller
than for the link $(r,s)$
\be
   |\theta_m - \theta_n| \le \theta_r - \theta_s,
\ee
see proposition \ref{prop:decay}. If $\theta_r - \theta_s = 1$ we thus know that
\be
   |F_{mn}| = |b_{mn} (\theta_m - \theta_n)| \le b_{mn}.
\ee
We thus obtain
\begin{align}
    & \frac{b_{rs}}{\eta(r,s)} =  \sum \nolimits_n F_{rn} \\
    & \mbox{such that} \;
       \theta_r - \theta_s = 1 \nn \\
    & \qquad \qquad F_{mn} = b_{mn} (\theta_m - \theta_n) \nn \\
    & \qquad \qquad |F_{mn}| \le b_{mn}  \; \forall \, \mbox{edges} \, (m,n) \nn \\
    & \qquad  \qquad \sum_{n=1}^N F_{mn} = 0 \; \forall \, m \neq r,s
\end{align}
Comparing to the expression (\ref{eq:maxflow}) for $\lambda_F(r,s)$ we see that two additional
constraints have to be satisfied. Additional constraints can only decrease the flow-value with respect
to the maximum in Equation (\ref{eq:maxflow}) such that we have
\begin{align}
    \frac{b_{rs}}{\eta(r,s)} \le \lambda_F(r,s)  \nn \\
    \Rightarrow \eta(r,s) \ge \frac{b_{rs}}{\lambda_F(r,s)} \, .
\end{align}

\end{proof}

\section{Proof of proposition \ref{prop:eta-upperbound}}
\label{sec:proof-eta-upperbound}

\begin{proof}
Consider first a reduced network consisting only of the link $(r,s)$ and the shortest alternative
path between the two nodes, which we denote as $(j_1=r, j_2, j_3, \cdots, j_n = s)$. Fixing
the nodal potentials such that $\psi_r - \psi_s = 1$, the direct flow over the link $(r,s)$ is given
by
\be
    F'_{r\rightarrow s} = b_{rs}
\ee
whereas the indirect flow over shortest alternative path is given by
\begin{align}
  F'_{r \rightarrow j_2} &= F'_{j_2 \rightarrow j_3} = \cdots = F'_{j_{n-1} \rightarrow j_n} \nn \\
  &= \left[ b_{j_1,j_2}^{-1} + b_{j_2,j_3}^{-1}  + \cdots + b_{j_{n-1},j_n}^{-1}  \right]^{-1} \nn \\
  &= \frac{1}{{\rm dist}_1^\mathrm{w}(r,s)}.
\end{align}
Reintroducing all edges to the grid can only increase the total flow from $r$ to $s$ such that
\be
    \sum_n F_{r \rightarrow n} \ge F'_{r\rightarrow s} + F'_{r\rightarrow j_2}
    = b_{rs} + \frac{1}{{\rm dist}_1^\mathrm{w}(r,s)} \, .
\ee
Thus we obtain (cf. Equation \ref{eq:eta-F-direct-total})
\be
    \eta(r,s) = \frac{F_{r \rightarrow s}}{\sum_n F_{r \rightarrow n}}
        \le \left[ 1 + \frac{1}{b_{rs} \times {\rm dist}_1^\mathrm{w}(r,s)} \right]^{-1}.
\ee
\end{proof}

\section{Proof of Proposition \ref{prop:decay}}
\label{sec:proof-decay}

In this appendix we first give the proof for proposition \ref{prop:decay} and then show
when the decay becomes strictly monotonous.

\begin{proof}
The proof is carried out by induction starting from $d = d_{\rm max}$. We only give
the proof for the maximum, the proof for the minimum proceeds in an analogous way.
We assume that the network is large enough such that $d_{\rm max} \ge 2$, otherwise
the statement is trivial anyway.

(1) Base case $d = d_{\rm max}$:
Consider the node $n$ of the network for which $\textrm{dist}(n,r) = d_\textrm{max}$
and $\psi_n$ assumes its maximum $\psi_n = u_{d_{\rm max}}$. By assumption
we have $\textrm{dist}(n,r) \ge 2$ such that the node $n$ cannot be adjacent to the
perturbed edge such that $q_n = 0$. The $n$-th component of Equation
(\ref{eq:Poisson}) yields
\begin{align}
   B_{nn}  \psi_n &= -  \sum \limits_{m \neq n} B_{nm} \psi_m  \nn \\
   &= -  \sum \limits_{\substack{m\neq n \\ \textrm{dist}(m,r) = d_{\rm max}}} B_{nm} \psi_m \nn \\
     & \qquad -  \sum \limits_{\substack{m \neq n \\ \textrm{dist}(m,r) = d_{\rm max}-1}} B_{nm} \psi_m  .
              \label{eqn:proof2-f_from_ A}
\end{align}
We define the abbreviations
\begin{align}
    \mathcal{B}_d =-  \sum \limits_{m \neq n, \textrm{dist}(m,r) = d} B_{nm}
\end{align}
and use some important properties of the matrix $\vec B$:
\begin{align}
   B_{nm} & \le 0  \; \textrm{for} \, n \neq m \qquad \Rightarrow \qquad \mathcal{B}_d \ge 0 \nn \\
   B_{nn} & \ge \mathcal{B}_{d_\textrm{max}} + \mathcal{B}_{d_\textrm{max}-1}    \, .
\end{align}
We can furthermore bound the values of $\psi_m$ in Equation (\ref{eqn:proof2-f_from_ A}) by
$u_{d_{\rm max}}$ or $u_{d_{\rm max}-1}$, respectively, such that we obtain
\begin{align}
   u_{d_{\rm max}} = \psi_n  & \le
       \frac{  \BB_{d_{\rm max}} u_{d_{\rm max}} + \BB_{d_{\rm max}-1} u_{d_{\rm max}-1} }{
                 \BB_{d_{\rm max}}  + \BB_{d_{\rm max}-1} } \nn \\
   \Rightarrow \; u_{d_{\rm max}} & \le u_{d_{\rm max}-1}, \nn \\
\end{align}

(2) Inductive step $d \rightarrow d-1$:
We consider the node $n$ of the network with $\mbox{dist}(n,r) = d$ and $\psi_n = u_d$.
Starting from Equation (\ref{eq:Poisson}) and using the same estimates as above,
we obtain
\begin{align}
     u_d = \psi_n  & = \frac{ q_n - \sum_{m \neq n} B_{nm} \psi_m }{B_{nn}} \nn \\
      & \le \frac{  \BB_{d-1} u_{d-1} + \BB_{d} u_{d} + \BB_{d+1} u_{d+1} }{
                 \BB_{d-1} + \BB_{d} + \BB_{d+1}} \, .\nn \\
       \label{eqn:proof2-f_from_ A2}
\end{align}
Note that the inhomogeneity $q_n \le 0$ for all nodes except for $n = r$.
With the induction hypothesis $u_{d+1} \le u_d$ this yields
\begin{align}
   u_{d}   & \le \frac{  \BB_{d-1} u_{d-1} + (\BB_{d} + \BB_{d+1}) u_{d} }{
                 \BB_{d-1} + \BB_{d} + \BB_{d+1}}  \nn \\
        \Rightarrow u_d & \le u_{d-1}.
\end{align}
which completes the proof.
\end{proof}

\section{Rerouting distance}
\label{app:redistance}

The rerouting distance introduced in definition \ref{def-reroute-dist} is a proper distance measure in the sense that it satisfies the axioms of a metric as shown in the following lemma. It can be calculated by mapping it to the two-edge disjoint shortest path problem, which can be solved by Suurballe's algorithm \cite{Suur84}. The mapping is provided by the lemma \ref{lem-2disjoint}.

\begin{lemma}
Consider an undirected graph with non-negative (all-equal) edge weights. Then the rerouting distance
${\rm edist}_{\rm re}^\textrm{w/u}[(r,s),(m,n)]$ of two edges $(r,s)$ and $(m,n)$ satisfies the following properties
\begin{enumerate}
\item
Positive definiteness: ${\rm edist}_{\rm re}^\textrm{w/u}[(r,s),(m,n)] \ge 0$.
\item
Symmetry: $${\rm edist}_{\rm re}^\textrm{w/u}[(r,s),(m,n)] = {\rm edist}_{\rm re}^\textrm{w/u}[(m,n),(r,s)]$$
\item
Triangular inequality:
\begin{align*}
 {\rm edist}_{\rm re}^\textrm{w/u}[(a,b),(r,s)]&\le {\rm edist}_{\rm re}^\textrm{w/u}[(a,b),(m,n)]\\
   & + {\rm edist}_{\rm re}^\textrm{w/u}[(m,n),(r,s)] 
\end{align*}
\end{enumerate}
both in the weighted and unweighted case.
\end{lemma}

\begin{proof}
(1) Positive definiteness:
As long as all edge weights are non-negative, all paths lengths and hence also
the rerouting distances are non-negative.

(2) Symmetry:
Suppose
\be
  \label{eq:repath-proof-metric-axioms}
  (v_0 = r, v_1, \ldots, v_i = m, v_{i+1} = n, \ldots v_k = s)
\ee
is the shortest rerouting path from $r$ to $s$ via $(m,n)$. Then
\be
  (v_{i+1} = n, v_{i+2}, \ldots, v_k = s, v_0 = r, v_1, \ldots, v_i = m)
\ee
is also a rerouting path from $n$ to $m$ via $(r,s)$. One can then show that this must be the shortest such rerouting path via contradiction. So suppose that
another path from $n$ to $m$ via $(r,s)$,
\be
  (u_{j+1} = n, u_{j+2}, \ldots u_\ell = s, u_0 = r, u_1, \ldots, u_j = m),
\ee
is shorter. Then the path
\be
  (v_0 = r, v_1, \ldots, v_i = m, v_{i+1} = n, \ldots v_\ell = s)
\ee
is a rerouting path from $r$ to $s$ via $(m,n)$ and it is shorter than than the one
defined in Equation~(\ref{eq:repath-proof-metric-axioms}).
This contradicts our initial assumption such that the path defined in
Equation~(\ref{eq:repath-proof-metric-axioms}) is the shortest rerouting path from
from $n$ to $m$ via $(r,s)$ and we obtain
\be
   {\rm edist}_{\rm re}^\textrm{w/u}[(r,s),(m,n)] = {\rm edist}_{\rm re}^\textrm{w/u}[(m,n),(r,s)].
\ee

(3) Triangle inequality: Let the paths
\begin{align*}
  p_1&=(v_0 = a, v_1, \ldots, v_i = m, v_{i+1} = n, \ldots v_\ell = b)\\
   p_2&=(u_0 = m, u_1, \ldots, u_j = r, u_{j+1} = s, \ldots u_\ell =n)
\end{align*}
be the shortest rerouting paths from $a$ to $b$ via edge $(m,n)$ and from $m$ to $n$ via $(r,s)$, respectively. Here, we assume the paths to be oriented as $v_i=m$, $v_{i+1}=n$ and $u_i = r, u_{i+1} = s$, but the proof is the same if the order of these vertices in the path is reversed. In addition to that, we assume the two distances on the right-hand-side of the inequality to be finite, otherwise the proof is trivial. We can extend the path $p_2$ to become a cycle by adding the edge $(n,m)$ to the end of the path
$$
  c_2=(u_0 = m, \ldots, u_i = r, u_{i+1} = s, \ldots u_\ell =n, u_{\ell+1}=m) .
$$
Now we can explicitly construct a rerouting path from $a$ to $b$ via $(r,s)$. Let $u_j\equiv v_p$ be the first vertex that appears in both $p_1$ and $c_2$ and let $u_k\equiv v_q$ the last such vertex. In this case, one of the following paths is a rerouting path from $a$ to $b$ via $(r,s)$
\begin{align*}
  p_3=(u_0 = a&, u_1, \ldots, u_j = v_p,\\
  &v_{p+1},\ldots ,v_{q-1},  u_{k} = v_q, \ldots u_\ell = b)\\
  \text{or } p_4=(u_0 = a&, u_1, \ldots, u_j = v_p,\\
  &v_{p-1},\ldots ,v_{q+1},  u_{k} = v_q, \ldots u_\ell = b).
\end{align*}
Assume without loss of generality that $p_3$ is a rerouting path from $a$ to $b$ via $(r,s)$. In this case, we obtain 
\begin{align*}
{\rm edist}&_{\rm re}^\textrm{w/u}[(a,b),(r,s)]\leq \text{length}((a,b))+\text{length}(p_3)\\
&\leq \text{length}((a,b))+\text{length}(p_1) + \text{length}(c_2)\\
&=\text{length}((a,b))+\text{length}(p_1)\\
&+\text{length}((m,n))+\text{length}(p_2)\\
&={\rm edist}_{\rm re}^\textrm{w/u}[(a,b),(m,n)]+{\rm edist}_{\rm re}^\textrm{w/u}[(m,n),(r,s)].
\end{align*}
Note that again the length of a path is the sum of the edge weights of all edges in the path when considering a weighted graph.
\end{proof}

\begin{lemma}
\label{lem-2disjoint}
The shortest rerouting path form $r$ to $s$ via edge $(m,n)$ is given by the
union of the edge $(m,n)$ and the two edge-\emph{independent} paths
$r \rightarrow m$ and $n \rightarrow s$ or
$r \rightarrow n$ and $m \rightarrow s$ which minimize the total path
length.
\end{lemma}

\begin{proof}
Assume that we have found a solution to the two-disjoint shortest path problem,
i.e. we have found two edge-independent paths
\begin{align}
   & (v_0 = r, v_1, \ldots, v_i = m) \\
   & (u_0 = s,  u_1 \ldots u_j = s),
   \label{eq:2-disjoint-proof}
\end{align}
which minimize the total path length.
By assumption the two paths are edge-independent such that
\be
   (v_0 = r, v_1, \ldots, v_i = m, u_0 = s,  u_1 \ldots u_j = s)
   \label{eq:re-from-two}
\ee
is a valid rerouting path.
Now it remains to show that this path is indeed the shortest possible. So assume
the contrary, i.e. that there exists a path
\be
  (w_0 = r, w_1, \ldots, w_i = m, w_{i+1} = n, \ldots w_k = s)
\ee
which is shorter than (\ref{eq:re-from-two}). But then the two paths
\begin{align}
   & (w_0 = r, w_1, \ldots, w_i = m) \\
   & (w_{i+1} = n,  w_{i+2} \ldots w_k = s)
\end{align}
are edge independent and have a shorter total path length than the two
paths \ref{eq:2-disjoint-proof}.
Contradiction.
\end{proof}

\bibliography{lodf}
\bibliographystyle{apsrev}

\end{document}